\documentclass[12pt]{article}

\usepackage{lmodern}
\usepackage[T1]{fontenc}

\usepackage{amsmath,amssymb,amsthm,fullpage,charter,verbatim,calc}
\usepackage{pifont,amsfonts,amsmath,amstext,amsthm,amssymb}
\usepackage{graphicx,color}
\usepackage{vmargin}
\usepackage{mathtools}
 \usepackage{MnSymbol}
\usepackage{abstract}
\setpapersize{A4}
\setmarginsrb{30mm}{20mm}{30mm}{20mm}{12pt}{11pt}{0pt}{11mm}
\usepackage{epic,eepic}
\newtheorem{theorem}{Theorem}

\newtheorem{corollary}{Corollary}

\newtheorem{proposition}{Proposition}
\newtheorem{definition}{Definition}

\newcommand*{\clNP}{$ \mathrm{\bf NP} $}
\newcommand*{\cH}{\mathcal{H}}
\newcommand*{\cC}{\mathcal{C}}
\newcommand*{\clAPX}{$\mathrm{\bf APX}$}
\title{\bf On Approximation Lower Bounds for TSP with Bounded Metrics}
\author{
Marek Karpinski\thanks{Dept. of Computer Science and the Hausdorff
    Center for Mathematics, University of Bonn.
    Supported in part by DFG grants and the Hausdorff Center grant EXC59-1.
    Email:~\texttt{marek@cs.uni-bonn.de}}
\and    
    Richard Schmied\thanks{Dept. of Computer Science, University of Bonn.
    Work supported by Hausdorff Doctoral Fellowship.
    Email:~\texttt{schmied@cs.uni-bonn.de}}
}
\date{}

\begin{document}
\maketitle

\begin{abstract}
We develop a new method for proving explicit approximation lower bounds for TSP
problems with bounded metrics improving on the best up to now known bounds. 
They almost match the best known bounds for unbounded metric TSP problems.
In particular, we prove the best known lower bound for TSP with bounded metrics 
 for the metric bound equal to $4$.
\end{abstract}

\section{Introduction}
We give first the basic definitions and an overview of the known 
results.

\subsubsection{Traveling Salesperson (TSP) Problem}
 We are given
 a metric space $(V,d)$ and 
 the task consists of  constructing a shortest tour visiting
 each vertex exactly once.

The TSP problem in metric spaces is one of the most fundamental \clNP-hard optimization
problems. 
 The decision version
of this problem was shown early to be \clNP-complete by Karp~\cite{K72}.
 Christofides~\cite{C76} gave 
 an  algorithm approximating the TSP problem within $3/2$, i.e., an algorithm
that produces a tour with length being
at most a factor $3/2$ from the optimum.
  
As for lower bounds, a reduction due to Papadimitriou
and Yannakakis~\cite{PY93}
and the PCP Theorem~\cite{ALM$^+$98}  together imply that there exists some constant, not better
than $1+ 10^{-6}$,  such that it is \clNP-hard
to approximate the TSP problem with distances  either one or two. 
For discussion of bounded metrics TSP, see also \cite{T00}.
This
hardness result was improved by Engebretsen~\cite{E03}, who proved that it is
 \clNP-hard to approximate the TSP problem restricted to distances one and two with an approximation
 factor better than  $5381/5380~(1.00018)$.
B\"ockenhauer and Seibert~\cite{BS00} studied the TSP problem with distances one, two
and three, and obtained an approximation lower bound of $3813/3812~(1.00026)$. 
Then, 
Papadimitriou and Vempala~\cite{PV06} showed that approximating  the general problem  
 with a constant approximation factor 
better than $220/219~(1.00456) $ is \clNP-hard.  Recently, this  bound was improved to
$ 185/184~(1.00543)$ by Lampis~\cite{L12}.

The restricted version of the TSP problem, in which the distance function takes values in 
$\{1,\ldots,B\} $, is referred to as the \emph{$(1,B)$--TSP problem}. 
 
The $(1,2)$--TSP problem can be approximated in polynomial time  with an approximation factor 
$8/7$ due to Berman and Karpinski \cite{BK06}.
 
 On the other hand, Engebretsen and Karpinski~\cite{EK06} proved that it
is \clNP-hard to approximate the $(1,B)$--TSP problem 
with an approximation factor less than  $741/740~(1.00135)$ for $B=2$ 
and $389/388~(1.00257)$ for $B=8$.

In this paper, we prove that the $(1,2)$--TSP and the $(1,4)$--TSP problem are \clNP-hard to approximate 
with an approximation factor less than $535/534$
and $337/336$, respectively.
\subsubsection{Asymmetric Traveling Salesperson (ATSP) Problem}
 We are given
 an asymmetric metric space $(V,d)$, i.e., $d$ is not necessarily symmetric, 
 and we would like to construct a shortest tour visiting
 every vertex exactly once.
 
 The best known algorithm for the ATSP problem 
 approximates the solution within $O(\,\log n/ \log \log n)$, where $n$ is the number of
vertices in the metric space~\cite{AGM$^+$10}.

On the other hand, Papadimitriou and Vempala~\cite{PV06}
proved that the ATSP problem is \clNP-hard to approximate with an approximation factor
less than $117/116~(1.00862) $.

It is conceivable that the special cases with bounded metric are easier to approximate
than the cases when the distance between two points grows  with the size of the 
instance. Clearly, the $(1,B)$--ATSP problem, in which  the distance function is taking
 values in the set $\{1,\ldots, B\}$, can be
approximated within $B$ by just picking any tour as the solution.

When we restrict the problem to distances one and two, it    
 can be approximated within $5/4$ due to Bl{\"a}ser \cite{B04}. 
 
 Furthermore, it is \clNP-hard to approximate this problem with an approximation factor 
 better than $321/320$ 
 ~ \cite{BK06}.
 
For the case $B=8$,  Engebretsen and Karpinski~\cite{EK06} constructed a reduction yielding 
the approximation lower bound $135/134$ for the  $(1,8)$--ATSP problem.

 In this paper, we prove that it is \clNP-hard to approximate the $(1,2)$--ATSP and the $(1,4)$--ATSP
 problem with an approximation factor less 
 than $207/206$ and $141/140$, respectively.

 
%
\subsubsection{Maximum Asymmetric Traveling Salesperson (MAX-ATSP) Problem}
We are given a complete directed graph $G$  and a weight function
$w$ assigning each edge of $G$ a nonnegative weight. The task is to find a 
tour of maximum weight visiting every vertex of $G$ exactly once .

This problem is  well-known and motivated by several applications (cf. \cite{BGS02}).
A good approximation algorithm for the  MAX--ATSP problem
yields a good approximation algorithm for many other optimization problems
such as the  Shortest Superstring problem, the Maximum Compression problem
and the $(1,2)$--ATSP problem. 

The MAX--$(0,1)$--ATSP problem is the restricted version of the MAX-ATSP problem, 
in which the weight function $w$ takes  values in the set $\{0,1\}$.

Vishwanathan~\cite{V92} constructed an approximation preserving reduction proving
 that any $(1/\alpha)$--approximation algorithm for the MAX--$(0,1)$--ATSP problem 
problem transforms in a $(2-\alpha)$-- approximation algorithm for the $(1,2)$--ATSP problem.
Due to this reduction, all negative results concerning the approximation of the $(1,2)$--ATSP
problem imply hardness results for the MAX--$(0,1)$--ATSP problem. 

Since the $(1,2)$--ATSP problem  is \clAPX-hard~\cite{PY93}, there is little hope for polynomial time
approximation algorithms with arbitrary good precision for the MAX--$(0,1)$--ATSP problem.
Due to  the explicit approximation lower bound for the $(1,2)$--ATSP problem given in~\cite{EK06},
 it is \clNP-hard to approximate the MAX--$(0,1)$--ATSP problem  with an approximation factor
less than 
$320/319$.

The best
known approximation algorithm for the
restricted version of this problem is due to Bl\"aser~\cite{B04} and 
achieves an  approximation ratio $5/4$.

For the general problem, Kaplan et al.~\cite{KLS$^+$05}  
designed an algorithm for the MAX--ATSP problem
  yielding the best known approximation upper bound of $3/2$.
  
On the approximation hardness side, Karpinski and Schmied~\cite{KS11}
constructed a reduction yielding the approximation lower bound $208/207$
for the 
MAX--ATSP problem.

In this paper, we prove that
approximating the MAX--$(0,1)$--ATSP problem with an approximation ratio less than 
$206/205$ is \clNP-hard.

\subsubsection{Overview of Known Explicit Approximation Lower Bounds and Our Results  }

\begin{figure}[h]
\begin{center}
\begin{tabular}{||c|c|c|c|c||}
\hline
\hline
 \centering \textbf{ $(1,B)$--ATSP problem}  &  $B=2$   & $B=4$   & $B=8$    &  unbounded    \\
  \hline
 Previously known  & $321/320$  &  $321/320$  
 & $135/134$ &  $117/116$\\
 results & $(1.00312)$  &  $(1.00312)$ 
 & $(1.00746)$ &  $(1.00862)$  \\
  & \cite{EK06}  &  \cite{EK06}  & \cite{EK06}  & \cite{PV06}   \\
 &     &       &     &           \\
 \hline
 Our results             & $207/206$ &  $141/140$ &  &   \\
              & $(1.00485)$ &  $(1.00714)$ &  &   \\
 \hline 
 \centering \textbf{ $(1,B)$--TSP problem}  &  $B=2$   & $B=4$   & $B=8$    &  unbounded    \\
 \hline
  Previously known  & $741/740$ &  $741/740$ & $389/388$ 
  &  $220/219 $  \\
  results & $(1.00135)$ &  $(1.00135)$ & $(1.00257)$ 
  &  $(1.00456) $  \\
   & \cite{EK06}  &  \cite{EK06}  & \cite{EK06}  &  \cite{PV06}  \\
 &     &       &     &           \\
 \hline
 Our results             & $535/534$ &  $337/336$ 
 & $337/336$ &   \\
              & $(1.00187)$ &  $(1.00297)$ 
 & $(1.00297)$ &   \\ 
 \hline 
  \hline
\end{tabular}
\end{center}
\caption{Known explicit approximation lower bounds and the new results.}
\label{fig:atspprobgadgets}
\end{figure}

 \begin{figure}[h]
\begin{center}
\begin{tabular}{||c|c|c|c|c||}
\hline
\hline
         &     \multicolumn{2}{c|}{\centering \textbf{MAX--$(0,1)$--ATSP problem}} 
         & \multicolumn{2}{c||}{\centering \textbf{MAX--ATSP problem} } \\
         \hline
Previously known  &   \multicolumn{2}{c|}{ $320/319$  }  & 
\multicolumn{2}{c||}{ $208/207$ } \\
 results &   \multicolumn{2}{c|}{ $(1.00314 )$  }  & 
\multicolumn{2}{c||}{ $(1.00483)$ } \\
 &     \multicolumn{2}{c|}{ \cite{EK06} }     &     \multicolumn{2}{c||}{ \cite{KS11} }          \\
 &     \multicolumn{2}{c|}{  }     &     \multicolumn{2}{c||}{  }          \\
 \hline
 Our results            &     \multicolumn{2}{c|}{ $206/205$ }   
 &   \multicolumn{2}{c||}{ $206/205$ }          \\
            &     \multicolumn{2}{c|}{ $(1.00487)$ }   
 &   \multicolumn{2}{c||}{ $(1.00487)$ }          \\
 \hline
 \hline
\end{tabular}
\end{center}
\caption{ Known explicit approximation lower bounds and the new results.}
\label{fig:maxatspresults}
\end{figure}

\section{Preliminaries}
In this section, we define the abbreviations and notations used in this paper.\\
\\ 
\noindent
Given a natural number $k$ and a finite set $V$, we use the abbreviation  $[k]$ for  $\{1,\ldots, k\}$
and ${V \choose 2}$ for the set $\{~S\subseteq V \mid |S|=2\}$.

Given an asymmetric metric space $(V,d)$ 
with $V=\{v_1,\ldots, v_n\}$ and $d:V\times V \rightarrow \mathbb{R}_{\geq0}$,
a \emph{Hamiltonian cycle} in $V$ or a \emph{tour} in $V$ 
is a cycle visiting each vertex $v_i$ in $V$ exactly once.

In the remainder, we specify a tour $\sigma$ in $V$ by 
$\sigma=\{(v_1,v_{i_1}), \ldots , (v_{i_{n-1}}, v_1)\}\subseteq V\times V$
or  explicitly by $\sigma=v_1\longrightarrow v_{i_1} \longrightarrow \cdots \longrightarrow 
v_{i_{n-1}} \longrightarrow v_1$.

Given a tour $\sigma$ in $(V,d)$, the length of a tour $\ell(\sigma)$ is defined 
by $$\ell(\sigma)=\sum\limits_{a\in \sigma}d(a).$$ 

Analogously, given a metric space $(V,d)$ with $V=\{v_1,\ldots, v_n\}$ and 
$d: {V \choose 2} \rightarrow \mathbb{R}_{>0}$, we specify a tour $\sigma$ in $V$
by $\sigma=\{~\{v_1,v_{i_1}\}, \ldots , \{v_{i_{n-1}}, v_1\}~\}\subseteq {V \choose 2}$
or 
explicitly by $\sigma=v_1 - v_{i_1} - \cdots - v_{i_{n-1}} - v_1$.

In order to specify, an instance $(V,d)$ of the $(1,2)$--ATSP problem, it suffices to 
identify the arcs $a\in V\times V$ with weight one. 
The same instance is specified by a directed graph $D=(V,A)$, where
$a\in A$ if and only if $d(a)=1$. Analogously, in the $(1,2)$--TSP problem, an instance is
completely specified by a graph $G=(V,E)$. 

In the remainder, we refer to an arc and an edge with weight $x\in \mathbb{N}$ as a \emph{$x$-arc}
and  \emph{$x$-edge}, respectively.

\section{Related Work}

\subsection{Hybrid Problem}

Berman and Karpinski~\cite{BK99}, see also \cite{BK01} and \cite{BK03}  introduced the following
Hybrid problem and proved that this problem is \clNP-hard to approximate
with some constant. 

\begin{definition}[Hybrid problem] 
Given a system of linear
equations mod 2 containing n variables, $m_2$ equations with exactly two variables, and $m_3$
equations with exactly three variables, find an assignment to the variables that satisfies as
many equations as possible.
\end{definition}
\noindent
In~\cite{BK99}, Berman and Karpinski constructed special instances of the Hybrid problem
with bounded occurrences of variables, for which 
 they  proved the following hardness result.
\begin{figure}[h]
\begin{center}
\fbox{ 
\input{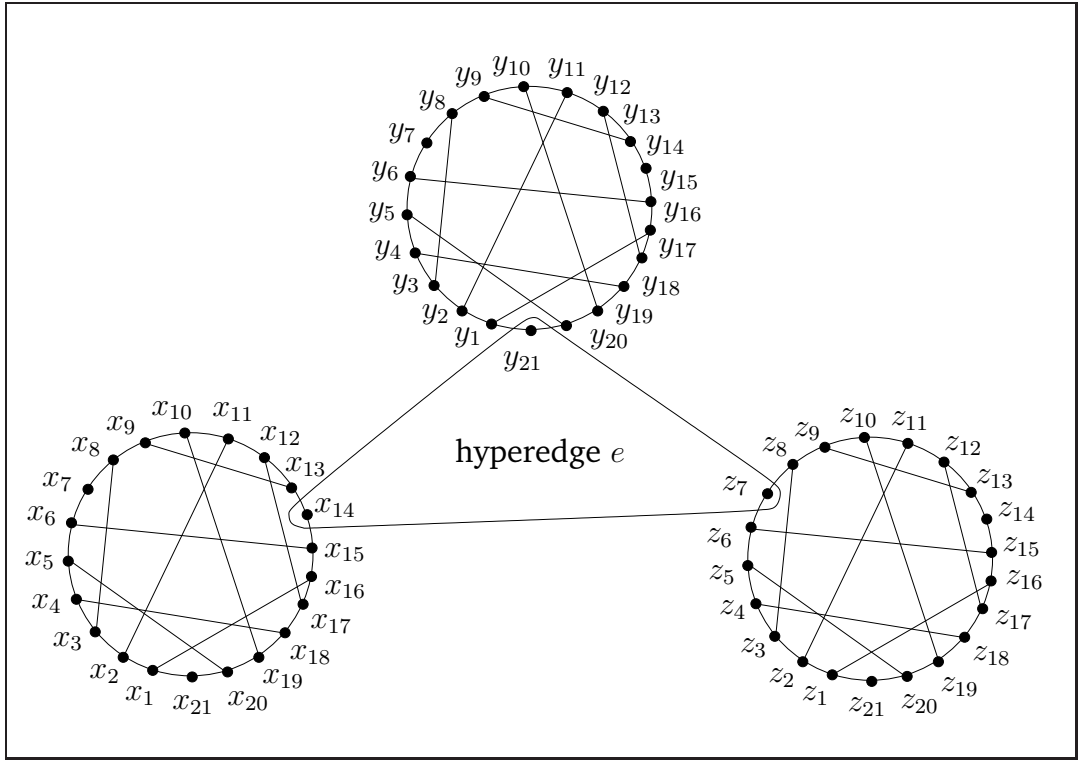}
}
\end{center}
\caption{ An example  of a Hybrid instance with circles 
$C^x$, $C^y$, $C^z$, and hyperedge $e=\{z_7,y_{21},x_{14}\}$.}
\label{fig:hybridssp}
\end{figure} 
\begin{theorem}[\cite{BK99}]\label{ssphybridsatz}
For any constant $\epsilon > 0$, there exists instances of the Hybrid problem with $42\nu$
variables, $60\nu$ equations with exactly two variables, and $2\nu$ equations with exactly three
variables such that:
\begin{enumerate}
\item[$(i)$] Each variable occurs exactly three times.
\item[$(ii)$] Either there is an assignment to the variables that leaves at most $\epsilon\nu$ equations unsatisfied,
or else every assignment to the variables leaves at least $(1-\epsilon)\nu$ equations
unsatisfied.
\item[$(iii)$] It is \clNP-hard to decide which of the two cases in item $(ii)$ above holds.
\end{enumerate}
\end{theorem}
\noindent
The instances of the Hybrid problem produced in Theorem~\ref{ssphybridsatz}
 have an even more special structure, which we are going to describe. 
 
 The equations
containing three variables are of the form $x \oplus y \oplus z = \{0, 1\}$. 
These equations  stem from the Theorem of H{\aa}stad~\cite{H01} dealing with the hardness
of approximating equations with exactly three variables. We refer to it as the  MAX--E3--LIN problem,
which can be seen as a special instance of the Hybrid problem.

\begin{theorem}[H{\aa}stad~\cite{H01}]\label{thm:max-e3-lin} 
For any constant $\delta\in (0,\frac{1}{2})$, there exists systems of linear equations
mod 2 with 2m equations and exactly three unknowns in each equation such that: 
\begin{enumerate}
\item[$(i)$]
 Each
variable in the instance occurs a constant number of times, half of them negated and half of
them unnegated. This constant grows as $\Omega(2^{1/\delta})$. 
\item[$(ii)$]
 Either there is an assignment satisfying
all but at most $\delta \cdot m$ equations, or every assignment leaves at least $(1-\delta)m$ equations
unsatisfied. 
\item[$(iii)$] It is \clNP-hard to distinguish between these two cases.
\end{enumerate}
\end{theorem}

For every variable $x$ of the original instance $\mathcal{E}_3$ of the MAX-E3-LIN problem, 
Berman and Karpinski introduced a corresponding set  
of variables $V_x$. If the variable $x$ occurs $t_x$ times in $\mathcal{E}_3$, then, $V_x$ contains
$7t_x$ variables $x_1,\ldots, x_{7t_x}$. 
The variables contained in
$Con(V_x)=\{x_i \mid i\in \{7\nu\mid \nu \in [t_x]\}  \}$ are called \emph{contact variables}, whereas
the variables in $C(V_x)=V_x\backslash Con(V_x)$ are called \emph{checker variables}.

All variables in $V_x$ are connected 
by equations of the form $x_i\oplus x_{i+1}=0$ with $i\in [7t_x-1]$ and $x_1\oplus x_{7t_x}=0$.
In addition to it, there exists equations of the form $x_i \oplus x_j =0$ with $\{i,j\}\in M^x$,
where  $M^x$ defines a perfect matching on the set
of checker variables. 
In the remainder, we refer 
to this construction as the circle $\cC^x$ containing the variables $x\in V_x$.

 Let $\mathcal{E}_3$ be an instance of the MAX--E3--LIN problem and $\cH$ be its corresponding instance of the
Hybrid problem. We denote by $V(\mathcal{E}_3)$ the set of variables which occur in the instance $\mathcal{E}_3$. 
Then, $\cH$ can be represented graphically by $|V(\mathcal{E}_3)|$ circles $\cC^x$ 
with $x\in V(\mathcal{E}_3)$ containing
the variables $V(\cC^x)=\{x_1, \ldots,x_{t_x}\}$ as vertices. 

The edges are identified
by the equations  included in  $\cH$. The equations with exactly three variables are represented by
hyperedges $e$ with cardinality $|e|=3$. 
The equations 
 $x_i\oplus x_{i+1}=0$ induce a cycle containing 
 the vertices $\{x_1, \ldots,x_{t_x}\}$ and the matching equations $x_i\oplus x_{j}=0$ with $\{i,j\}\in M^x$
  induce a perfect matching on the set of checker variables. 
  An example of an instance of the Hybrid problem is depicted in Figure~\ref{fig:hybridssp}.
   
  In summary, we notice that
 there are four type of equations in 
the Hybrid problem $(i)$ the circle equations $x_i\oplus x_{i+1}=0$ with $i\in [7t_x-1]$,  
$(ii)$ circle border equations $x_1\oplus x_{7t_x}$, $(iii)$ matching equations
$x_i\oplus x_j=0$ with $\{i,j\}\in M^x$, and $(iv)$ equations with three variables 
of the form $x\oplus y \oplus z=\{0,1\}$.\\
In the remainder, we may assume that equations with three variables are of the form 
 $x \oplus y \oplus z = 0$ or $\bar{x} \oplus y \oplus z = 0$  due to
 the transformation $\bar{x} \oplus y \oplus z = 0 \equiv x \oplus y \oplus z = 1$.

\subsection{ Approximation Hardness of TSP Problems}
\subsubsection{Reducing the Hybrid Problem to the $(1,2)$--(A)TSP Problem }
\begin{figure}[h]
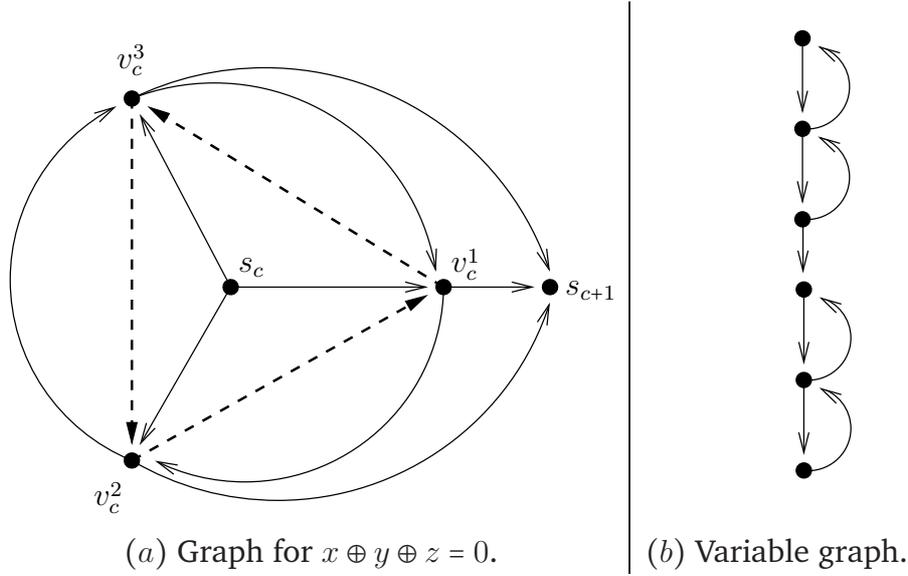

\begin{center}
\begin{tabular}[c]{c|c}
\input{figures/fig12atspgadvar3.pspdftex}& \input{figures/fig12atspek1.pspdftex}\\
 $(a)$ Graph for $x \oplus y \oplus z =0$. & $(b)$ Variable graph.
 \end{tabular}
 \end{center}
\caption{Gadgets used in \cite{EK06}.}
\label{fig:atspprobgadgets}
\end{figure}
Engebretsen and Karpinski~\cite{EK06} constructed an approximation preserving 
reduction from the Hybrid problem
to the $(1,2)$--ATSP problem to prove explicit approximation lower bounds for the 
latter problem. They introduced graphs (gadgets), which simulate variables, equations with
two variables and equations with three variables.  In particular, the graphs corresponding  
 to equations of the form $x \oplus y \oplus z =0 $ and to variables are displayed in Figure~\ref{fig:atspprobgadgets}    
$(a)$ and $(b)$, respectively.

For the graph corresponding to an equation with three variables, they proved the following 
statement.

\begin{proposition}[\cite{EK06}]\label{pro:gadget3atsp}   
There is a Hamiltonian path from $s_c$ to $s_{c+1}$ in Figure~\ref{fig:atspprobgadgets} $(a)$ if and
only if an even number of ticked edges is traversed. 
\end{proposition}

\begin{figure}[h]
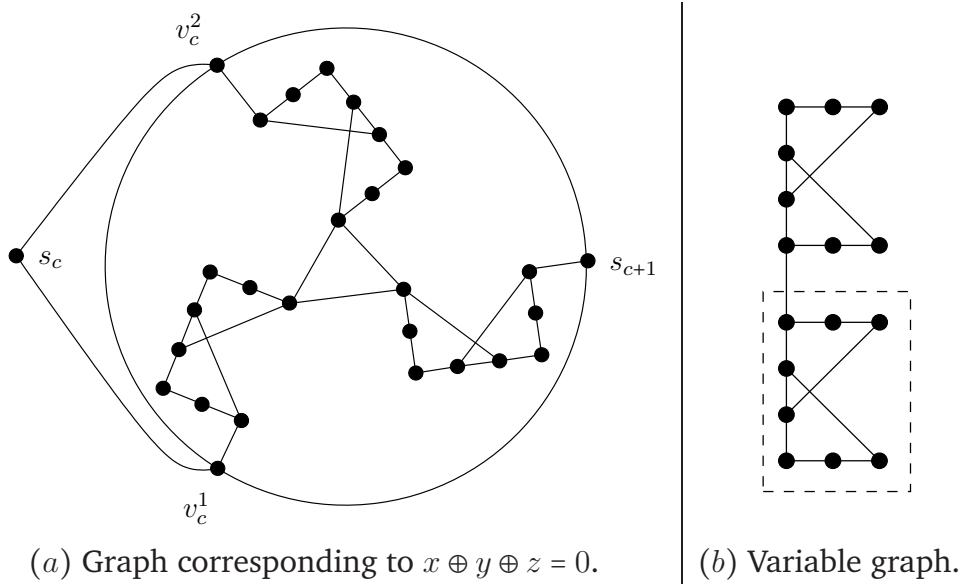

\begin{center}
\begin{tabular}[c]{c|c}
\input{figures/fig12tsp2.pspdftex}& \input{figures/fig12tspvargad.pspdftex}\\
 $(a)$ Graph corresponding to $x\oplus y \oplus z= 0$. & $(b)$ Variable graph.
 \end{tabular}
 \end{center}
\caption{Gadgets used in \cite{EK06} to prove approximation hardness of the $(1,2)$--TSP problem.}
\label{fig:tspprobgadgets}
\end{figure}

A similar reduction was constructed in order to prove explicit approximation lower bounds for the
$(1,2)$--TSP problem. The corresponding graphs are depicted in Figure~\ref{fig:tspprobgadgets}. 
The graph contained in the dashed box in Figure~\ref{fig:tspprobgadgets} $(b)$ 
will  play a crucial role in our reduction and  we refer to it as \emph{parity graph}.
In particular, our variable gadget consists only of a parity graph.

In the reduction of the $(1,2)$--TSP problem,
 the following statement was proved   for the 
 graph corresponding to equations of the form $x\oplus y \oplus z =0$.

\begin{proposition}[\cite{EK06}]\label{pro:gadget3tsp} 
There is a simple path from    $s_c$ to  $s_{c+1}$ in Figure~\ref{fig:tspprobgadgets} $(a)$
 containing the vertices 
$v\in \{v^1_c, v^2_c\}$  if and
only if an even number of parity graphs is  traversed. 
\end{proposition}

The former mentioned reductions combined with Theorem~\ref{ssphybridsatz} yield the following explicit 
approximation lower bounds.

\begin{theorem}[\cite{EK06}]\label{thm:ek6:12atsp}
It is \clNP-hard to approximate 
the $(1,2)$--ATSP and the $(1,2)$--TSP problem with an approximation ratio less than 
$321/320~(1.00312)$
and $741/740~(1.00135)$, respectively.
\end{theorem}

\subsubsection{Explicit Approximation Lower Bounds for the MAX--ATSP problem}

By replacing all edges with weight
two of an instance of the $(1,2)$--ATSP problem by edges of weight zero, 
we obtain an instance of the MAX--$(0,1)$--ATSP problem, which  relates the   
$(1,2)$--ATSP problem to the MAX--ATSP problem in the following sense.
\begin{theorem}[\cite{V92}]\label{thm:atsp-maxatsp}
An $( 1/\alpha)$--approximation algorithm for the MAX--$(0,1)$--ATSP problem implies
an $(2-\alpha)$--approximation algorithm for the $(1,2)$--ATSP problem. 
\end{theorem}   
\noindent
This reduction transforms  every hardness result addressing the $(1,2)$--ATSP problem
 into a hardness result for the MAX--$(0,1)$--ATSP problem.
In particular, Theorem~\ref{thm:ek6:12atsp} implies the best known  explicit 
approximation lower bound for the MAX--$(0,1)$--ATSP problem.
\begin{corollary}
It is \clNP-hard to approximate 
the MAX--$(0,1)$--ATSP problem within any better than  $320/319~(1.00314)$.
\end{corollary}

\section{Our Contribution}

We now formulate our main results.
\begin{theorem}\label{thm:main12atsp}
Suppose we are given an instance $\cH$ of the Hybrid problem  with $n$ circles,
$m_2$ equations with two variables and  $m_3$ equations with exactly
three variables with the properties described in Theorem~\ref{ssphybridsatz}.\\
\begin{enumerate}
\item[$(i)$]  Then, it is possible to
construct in polynomial time an instance $D_{\cH}$ of
the $(1,2)$--ATSP problem with the following properties: 
\begin{enumerate}
\item[$(a)$] If there exists an assignment $\phi$ to the variables of $\cH$ which leaves at most  $u$ equations
unsatisfied, then, there exist a tour $\sigma_{\phi}$ in $D_{\cH}$  with  length
at most $\ell(\sigma_{\phi})=3m_2+13 m_3+n+1+u$.   
\item[$(b)$] From every tour $\sigma$ in  $D_{\cH}$  with length 
$\ell(\sigma)=3m_2+13 m_3+n+1+u$,
we can construct in polynomial time  
an assignment $\psi_{\sigma}$ to the variables
of $\cH$ that leaves at most $u$ equations in $\cH$ unsatisfied.
\end{enumerate}
\item[$(ii)$]  Furthermore, it is possible to
construct in polynomial time an instance $(V_{\cH}, d_{\cH})$ of
the $(1,4)$--ATSP problem with the following properties: 
\begin{enumerate}
\item[$(a)$] If there exists an assignment $\phi$ to the variables of $\cH$ which leaves at most  $u$ equations
unsatisfied, then, there exist a tour $\sigma_{\phi}$ in $(V_{\cH}, d_{\cH})$  with  length
at most $\ell(\sigma_{\phi})=4m_2+20 m_3+2n+2u+2$.   
\item[$(b)$] From every tour $\sigma$ in  $(V_{\cH}, d_{\cH})$  with length 
$\ell(\sigma)=4m_2+20 m_3+2n+2u+2$,
we can construct in polynomial time  
an assignment $\psi_{\sigma}$ to the variables
of $\cH$ that leaves at most $u$ equations in $\cH$ unsatisfied.
\end{enumerate}
\end{enumerate}
\end{theorem}
\noindent
The former theorem can be used to derive an explicit approximation lower bound for
the $(1,2)$--ATSP problem by reducing instances of the Hybrid problem
of the form described in Theorem~\ref{ssphybridsatz} to the $(1,2)$--ATSP problem.
\begin{corollary}\label{coro:explatsp}
It is  \clNP-hard to approximate  the $(1,2)$--ATSP problem
with an approximation factor less than $207/206~(1.00485)$.
\end{corollary}
\begin{proof}
Let $\mathcal{E}_3$ be an instance of the MAX--E3--LIN problem.
We define $k$ to  
be the minimum number of occurences of a variable in $\mathcal{E}_3$.
According to Theorem~\ref{thm:max-e3-lin}, we may  
choose   $\delta >0$ such that  
$\frac{207-\delta}{206+\delta +\frac {7}k}\geq
\frac{207}{206}-\epsilon$ holds. Given an instance $\mathcal{E}_3$ of the MAX--E3--LIN problem
with $\delta'\in (0,\delta)$, we
 generate  the corresponding instance 
 $\cH$ of the Hybrid problem. Then, 
we construct the corresponding instance $D_{\cH}$ of the $(1,2)$-ATSP problem
 with the properties described in Theorem~\ref{thm:main12atsp}.
We conclude according to Theorem~\ref{ssphybridsatz} that there exist a tour 
in  $D_{\cH}$ 
with length at most 
$$3\cdot60\nu +13\cdot2\nu +\delta'\nu + n+1 \leq(206+\delta'+\frac {n+1}{\nu} )\nu  \leq(206+\delta'+\frac { 6+1}k )\nu $$ 
or the length of a tour in  $D_{\cH}$  is bounded from below by  
$$3\cdot60\nu +13\cdot2\nu +(1-\delta')\nu + n +1\geq (206+(1-\delta') )\nu  \geq (207-\delta')\nu  .$$
From Theorem~\ref{ssphybridsatz}, we know that the two cases above are \clNP-hard
to distinguish.
Hence, for every $\epsilon>0$, it is  \clNP-hard to find a solution to
 the  $(1,2)$-ATSP problem with an approximation ratio
 $\frac{207-\delta'}{206+\delta' +\frac {7}k}\geq
\frac{207}{206}-\epsilon$.
\end{proof}
Analogously, we derive the following statement.
\begin{corollary}
It is  \clNP-hard to approximate  the $(1,4)$--ATSP problem
with an approximation factor less than $141/140~(1.00714)$.
\end{corollary}
\noindent 
For the symmetric version of the problems, we construct reductions from the Hybrid 
problem with similar properties.
\begin{theorem}\label{thm:main12tsp}
Suppose we are given an instance $\cH$ of the Hybrid problem  with $n$ circles,
$m_2$ equations with two variables and  $m_3$ equations with exactly
three variables with the properties described in Theorem~\ref{ssphybridsatz}.
\begin{enumerate}
\item[$(i)$] 
 Then, it is possible to
construct in polynomial time an instance $G_{\cH}$ of
the $(1,2)$--TSP problem with the following properties: 
\begin{enumerate}
\item[$(a)$] If there exists an assignment $\phi$ to the variables of $\cH$ which leaves at most  $u$ equations
unsatisfied, then, there exist a tour $\sigma_{\phi}$ in $G_{\cH}$  with  length
at most $\ell(\sigma_{\phi})=8m_2+27 m_3+3n+1+u$.   
\item[$(b)$] From every tour $\sigma$ in  $G_{\cH}$  with length 
$\ell(\sigma)=8m_2+27 m_3+3n+1+u$,
we can construct in polynomial time  
an assignment $\psi_{\sigma}$ to the variables
of $\cH$ that leaves at most $u$ equations in $\cH$ unsatisfied.
\end{enumerate}
\item[$(ii)$] 
Furthermore, it is possible to
construct in polynomial time an instance $(V_{\cH}, d_{\cH})$ of
the $(1,4)$--TSP problem with the following properties: 
\begin{enumerate}
\item[$(a)$] If there exists an assignment $\phi$ to the variables of $\cH$ which leaves at most  $u$ equations
unsatisfied, then, there exist a tour $\sigma_{\phi}$ in $(V_{\cH}, d_{\cH})$  with  length
at most $\ell(\sigma_{\phi})=10m_2+36 m_3+6n+2+2u$.   
\item[$(b)$] From every tour $\sigma$ in  $(V_{\cH}, d_{\cH})$  with length 
$\ell(\sigma)=10m_2+36 m_3+6n+2+2u$,
we can construct in polynomial time  
an assignment $\psi_{\sigma}$ to the variables
of $\cH$ that leaves at most $u$ equations in $\cH$ unsatisfied.
\end{enumerate}
\end{enumerate}
\end{theorem}
\noindent
Analogously, we combine the former theorem with the explicit approximation lower bound for
 the Hybrid problem
of the form described in Theorem~\ref{ssphybridsatz} yielding the following
approximation hardness result.
\begin{corollary}
It is  \clNP-hard to approximate  the $(1,2)$--TSP and 
the $(1,4)$--TSP problem
with an approximation factor less than $535/534~( 1.00187)$
and $337/336~( 1.00297)$, respectively.
\end{corollary}
\noindent
From Theorem~\ref{thm:atsp-maxatsp} and Corollary~\ref{coro:explatsp}, we obtain the following 
explicit approximation lower bound.
\begin{corollary}
It is  \clNP-hard to approximate  the MAX--$(0,1)$--ATSP problem   
with an approximation factor less than $206/205~( 1.00487)$.
\end{corollary}

\section{Approximation Hardness of the  $(1,2)$--ATSP problem}\label{sec:12atsp}
Before we proceed to the proof of Theorem~\ref{thm:main12atsp} $(i)$, we describe the 
reduction from a high-level view and try to build some intuition.
\subsection{Main Ideas}
As mentioned above, we prove our hardness results by a reduction from the Hybrid 
problem. Let $\mathcal{E}_3$ be an instance of the MAX-E3-LIN problem
and $\cH$  the corresponding instance  of the Hybrid problem.
Every variable $x^l$ in the original instance $\mathcal{E}_3$ introduces an
associated circle $\cC_l$ in the instance $\cH$ as illustrated in Figure~\ref{fig:hybridssp}.

\begin{figure}[h]
\begin{center}
\input{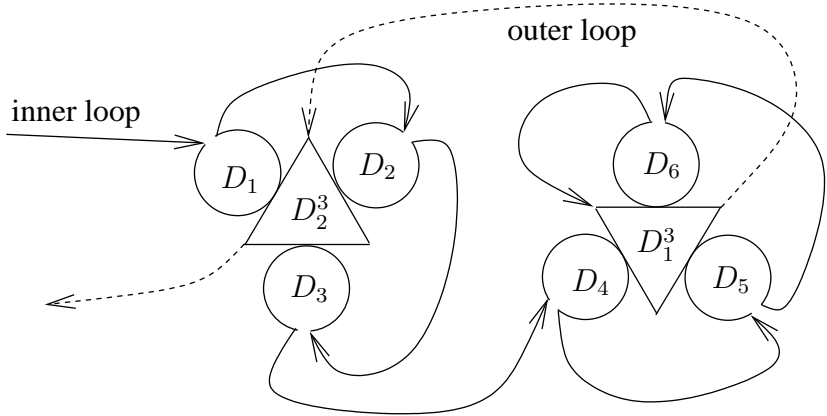}
\end{center}
\caption{An illustration of the instance $D_{\cH}$ and a tour in $D_{\cH}$.}
\label{figwholeconstr}
\end{figure} 

The main idea of our reduction is to make use of the special structure of the circles
in $\cH$. Every circle $\cC_l$ in $\cH$ corresponds to a graph $D_l$ in the  
instance $D_{\cH}$ of the $(1,2)$-ATSP problem. Moreover, $D_l$ is a subgraph
of $D_{\cH}$, which builds almost a cycle. An assignment to
the variable $x^l$ will have a natural interpretation in this reduction. The parity of $x_l$ 
corresponds to the direction
of movement in $D_l$ of the underlying tour.

The circle graphs of $D_{\cH}$ are connected
and build together the \emph{inner loop} of $D_{\cH}$. Every variable $x^l_i$ in a circle 
$\cC_l$ possesses an associated parity graph $P^l_i$ (Figure~\ref{fig:12atsp:1}) in $D_l$ as a subgraph.
The two natural ways to traverse a parity graph will be called $0/1$-traversals   
and correspond to the parity of the variable $x^l_i$. Some of the parity graphs in $D_l$ 
are also contained in graphs $D^3_c$ (Figure~\ref{fig:atspprobgadgets}$(a)$ and 
Figure~\ref{fig:12atsp:3} for a more detailed view) 
corresponding to equations with three variables of the
form $g^3_c \equiv x\oplus y \oplus z =0 $.

These graphs are connected and build the 
\emph{outer loop} of $D_{\cH}$. The whole construction is illustrated in Figure~\ref{figwholeconstr}.
The outer loop of the tour checks whether the $0/1$-traversals of the parity graphs  
correspond to an satisfying assignment of the equations with three variables.
If an underlying equation is not satisfied by the assignment defined via
$0/1$-traversal of the associated parity graph, it will be punished by
using a costly $2$-arc.

\subsection{ Constructing $D_{\cH}$ from a Hybrid Instance $\cH$}\label{sec:constrinst(12)atsp}
Given a instance of the Hybrid problem $\cH$, we are going to construct the corresponding
 instance $D_{\cH}=(V(D_{\cH}),A(D_{\cH}))$ of the $(1,2)$-ATSP problem.\\
 \\
 For every type of equation in $\cH$, we will introduce a specific  graph or
 a specific way to connect the so far constructed subgraphs. In particular, we will distinguish 
 between graphs corresponding to circle equations, matching equations, circle border equations and  
 equations with three variables. First of all, we introduce graphs corresponding to the variables in 
 $\cH$.
 
 \subsubsection{Variable Graphs}
 Let $\cH$ be an instance of the hybrid problem and $\cC_l$ a circle in $\cH$. 
  For every variable $x^l_i$ in the circle $\cC_l$, 
 we introduce the parity graph $P^l_i$ consisting of the vertices $v^{l1}_i$, $v^{l\bot}_i$ and $v^{l0}_i$
 depicted in Figure~\ref{fig:12atsp:1}.
 \begin{figure}[h!]
\begin{center}
 \input{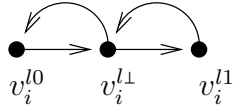}
\end{center}
\caption{ Parity graph $P^l_i$ corresponding to the variable $x^l_i$ in circle $\cC_l$.}
\label{fig:12atsp:1}
\end{figure} 
%
%
\subsubsection{Matching and Circle Equations }
 Let $\cH$ be an instance of the hybrid problem, $\cC_l$ a circle in $\cH$
 and $M_l$ the associated perfect matching. Furthermore, let $x^l_i\oplus x^l_j=0$
 with $e=\{i,j\}\in M_l$ and $i<j$ be a matching equation. Due to the construction of 
 $\cH$, the circle equations $x^l_i\oplus x^l_{i+1}=0$ and $x^l_j\oplus x^l_{j+1}=0$ are both
 contained in $\cC_l$. 
 \begin{figure}[h]
\begin{center}
 \input{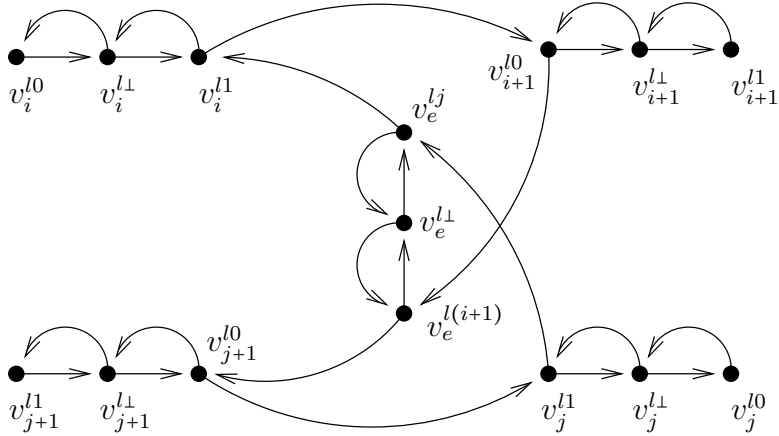}
\end{center}
\caption{ Connecting the parity graph $P^l_e$}
\label{fig:12atsp:2}
\end{figure} 
Then, we introduce the associated 
 parity graph $P^l_e$ consisting of the vertices $v^{lj}_{e}$, $v^{l\bot}_{e}$ and $v^{l(i+1)}_{e}$.  
 In addition to it, we connect the parity graphs $P^l_i$, $P^l_{i+1}$, $P^l_{j}$, $P^l_{j+1}$
 and $P^l_e$ as depicted in Figure~\ref{fig:12atsp:2}.
\noindent
\subsubsection{ Graphs Corresponding to Equations with Three Variables}
 Let $g^3_c\equiv x^l_i\oplus x^s_j \oplus x^k_t =0 $ be an equation with three variables in 
 $\cH$. Then, we introduce the  graph $D^3_c$ (Figure~\ref{fig:atspprobgadgets}) 
 corresponding to the equation $g^3_c$.
 The graph $D^3_c$ includes the vertices $s_{c}$, $v^1_c$, $v^2_c$, $v^3_c$ and $s_{c+1}$.
 \begin{figure}[h]
\begin{center}
 \input{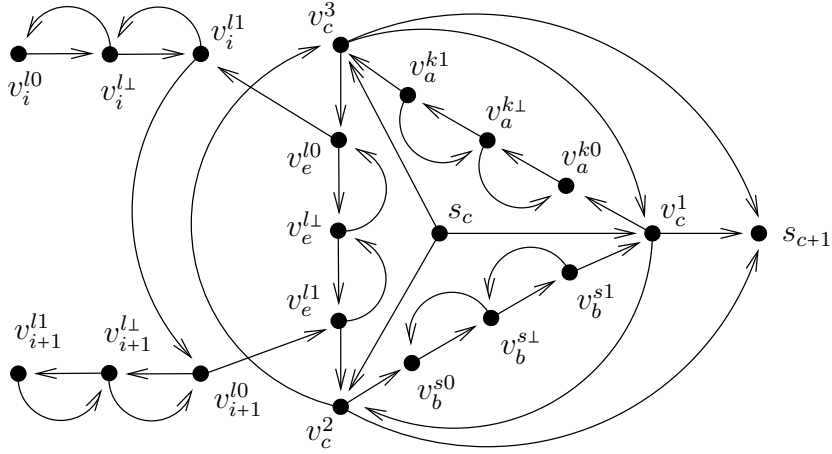} 
\end{center}
\caption{The graph $D^3_c$ corresponding to $g^3_c \equiv  x^l_i\oplus x^s_j \oplus x^u_k =0 $ connected
to graphs corresponding to $x^l_{i} \oplus x^l_{i+1} =0$.}
\label{fig:12atsp:3}
\end{figure} 
Furthermore, it contains the parity graphs $P^l_{e}$, $P^s_{b}$ and $P^k_{a}$ as subgraphs, where
 $e=\{i,i+1\}$, $b=\{j,j+1\}$ and $a=\{t,t+1\}$. Exemplary, we display $D^3_c$ with its connections
 to the graph corresponding to the circle equation $x^l_{i} \oplus x^l_{i+1} =0$ in Figure~\ref{fig:12atsp:3}.
 
 In case of $g^3_c\equiv \bar{x}^l_i\oplus x^s_j \oplus x^u_k =0$, we connect the parity graphs with 
 arcs $(v^{l1}_{i},v^{l0}_{e} )$, $(v^{l0}_{i+1},v^{l1}_{i})$ and $(v^{l1}_{e},v^{l0}_{i+1})$.
\subsubsection{Graphs Corresponding to Circle Border  Equations } 
\begin{figure}[h]
\begin{center}
 \input{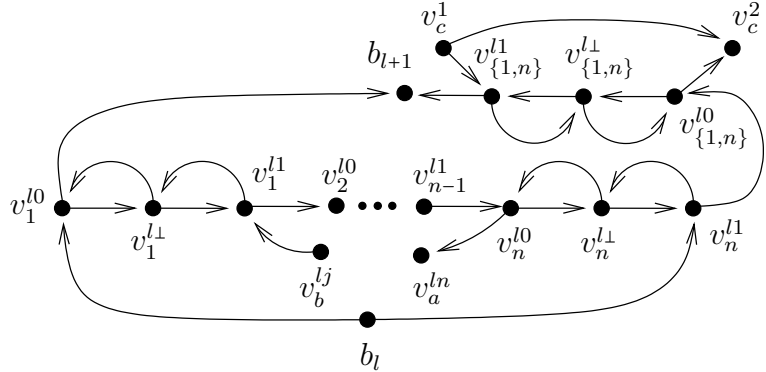}
\end{center}
\caption{The graph  corresponding to $ x^l_1\oplus x^l_n =0 $ }
\label{fig:12atspcb}
\end{figure}  
Let $\cC_l$ and $\cC_{l+1}$ be circles in $\cH$. In addition, let $x^l_1\oplus x^l_n=0$ be the 
 circle border equation of $\cC_l$. Then, we introduce the vertex $b_l$ and  connect it
  to $v^{l0}_1$ and $v^{l1}_n$. Let $b_{l+1}$ be the  vertex corresponding to the circle
 $C_{l+1}$. We draw an arc from $v^{l0}_1$ to $b_{l+1}$. Finally, we connect the vertex 
 $v^{l0}_{\{n,1\}}$ to $b_{l+1}$. Recall that $x_n$ also occurs in the equation $g^3_c$, which is 
 an equation with three variables in $\cH$.
 This construction is illustrated in Figure~\ref{fig:12atspcb}, where only a part of the corresponding graph $D^3_c$
 is depicted. 
 
 Let $\cC_n$ be the last circle in $\cH$. Then, we set $b_{n+1}=s_1$ as $s_1$ is the starting vertex 
 of the graph $D^3_1$ corresponding to the equation $g^3_1$. This is the whole description of the graph $D_{\cH}$.
 
 Next, we are going to describe the associated tour $\sigma_{\phi}$ in $D_{\cH}$ given an assignment to the variables 
 in $\cH$. 
 %
%
 %
 %
%
\subsection{Constructing the Tour $\sigma_{\phi}$ from an Assignment $\phi$   }
Let $\cH$ be an instance of the Hybrid problem and  $D_{\cH}=(V_{\cH},A_{\cH})$ the corresponding
instance of the $(1,2)$-ATSP problem as defined in Section~\ref{sec:constrinst(12)atsp}.
Given an assignment $\phi: V(\cH)\rightarrow \{1,0\}$ to the variables of $\cH$, 
we are going to construct the associated Hamiltonian tour $\sigma_{\phi}$ in $D_{\cH}$. 
In addition to it, we analyze the relation between the length of the tour $\sigma_{\phi}$ and
the number of satisfied equations by $\phi$.\\
\\
Let $\cH$ be an instance of the Hybrid problem consisting of circles $\cC_1, \cC_2, \ldots , \cC_m$ and 
equations with three variables $g^3_j$ with $j \in [m_3]$.  
The associated Hamiltonian tour $\sigma_{\phi}$ in $D_{\cH}$ starts at the vertex $b_1$. 
From a high level view,
$\sigma_{\phi}$ traverses all graphs corresponding to the equations associated with the circle $\cC_1$
ending with the vertex $b_2$. Successively, it passes all graphs for each circle in $\cH$ until
it reaches the vertex $b_m=s_1$ as  $s_1$ is the starting vertex of the graph $D^3_1$. 

At this point, the tour begins to traverse the remaining graphs $D^3_c$
which are simulating the equations with three variables in $\cH$.  
By now, some of the parity graphs appearing in  graphs $D^3_c$ already have been traversed
in the \emph{inner loop} of $\sigma_{\phi}$. 
The \emph{outer loop} checks whether for each graph $D^3_c$, an even number of parity graphs has been
traversed in the inner loop. In every situation, in which $\phi$ does not satisfy the underlying equation,
the tour needs to use a  $2$-arc. This paths, consisting of  $1$-arcs, will be aligned by means of
$2$-arcs in order to build a Hamiltonian tour in $D_{\cH}$. For each circle $\cC_l$, we will 
have to use $2$-arcs in order to obtain 
a Hamiltonian path from $b_l$ to $b_{l+1}$ traversing all graphs associated with $\cC_l$ in some order
except in the case when all variables in the circle have the same parity.
Since we specify only a part of  the tour $\sigma_{\phi}$ we rather refer to a representative tour
from a set of tours having the same length and the same specification.     

In order to analyze the length of the tour in relation to the 
number of satisfied equation, we are going to examine the part of $\sigma_{\phi}$ passing
 the  graphs corresponding to the underlying equation and account the local 
length to the analyzed parts of the tour. Let us begin to describe $\sigma_{\phi}$ passing through 
parity graphs associated to variables 
in $\cH$.
\begin{figure}[h]
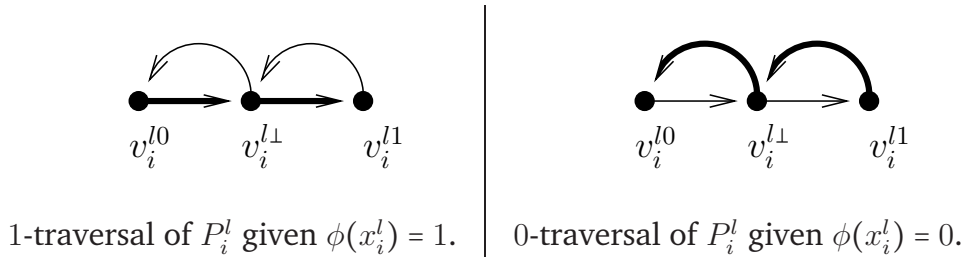

\begin{center}
\begin{tabular}[c]{c|c}
 \input{figures/fig12atsp4a.pspdftex} & \input{figures/fig12atsp4b.pspdftex} \\
 $1$-traversal of $P^l_i$ given $\phi(x^l_i)=1$.   & $0$-traversal of $P^l_i$ given $\phi(x^l_i)=0$.
 \end{tabular}
\end{center}
\caption{Traversal of the graph $P^l_i$ given the assignment $\phi$. }
\label{fig:12atsp:4}
\end{figure}   

\subsubsection{Traversing Parity Graphs}
Let $x^l_i$ be a variable in $\cH$.
Then, the tour $\sigma_{\phi}$ traverses the 
parity graph $P^l_i$ using the path  $v^{l[1-\phi(x^l_i)]}_i \rightarrow v^{l\bot}_i \rightarrow v^{l\phi(x^l_i)}_i$.
In the remainder, we call this part of the tour a \emph{$\phi(x^l_i)$-traversal} of the parity graph. 
In 
Figure~\ref{fig:12atsp:4}, we depicted the corresponding traversals of the graph $P^l_i$ 
given the assignment $\phi(x^l_i)$, 
where the traversed arcs are illustrated by thick arrows. In both cases, we associate the local length 
$2$ with this part of the tour.
\begin{figure}[h!]
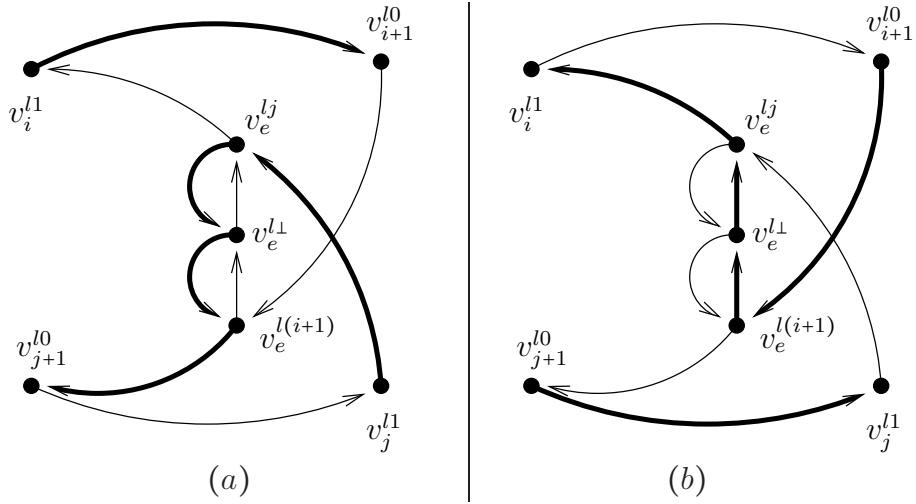

\begin{center}
\begin{tabular}[c]{c|c}
 \input{figures/fig12atsp5.pspdftex}& \input{figures/fig12atsp6.pspdftex}\\
 $(a)$ & $(b)$
 \end{tabular}
\end{center}
\caption{1. Case  $\phi(x_i) \oplus \phi(x_{i+1})=0$, 
$\phi(x_i) \oplus \phi(x_{j})=0$
and $\phi(x_j) \oplus \phi(x_{j+1})=0$. }
\label{fig:12atsp:5}
\end{figure} 
\begin{figure}[h!]
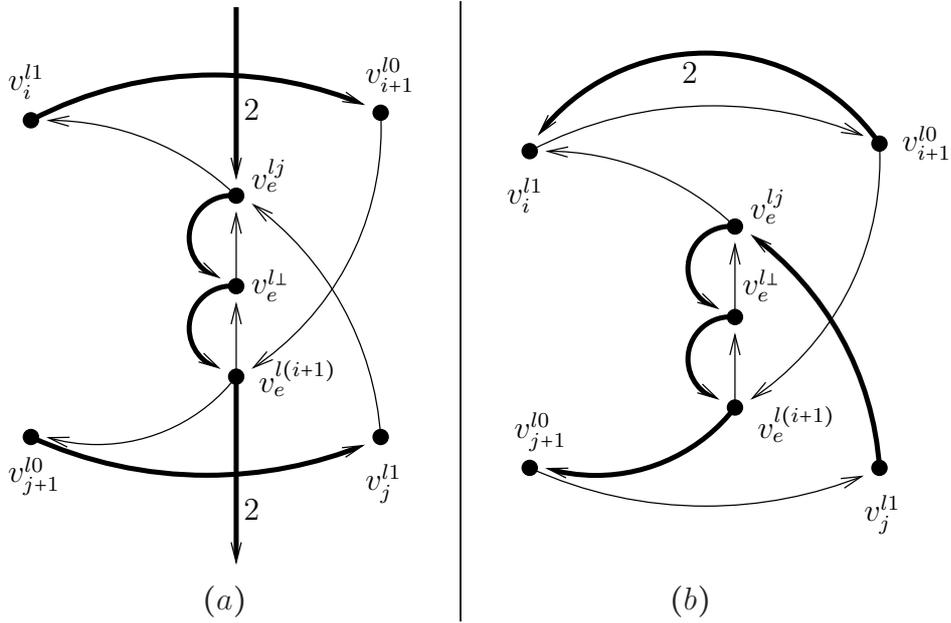

\begin{center}
\begin{tabular}[c]{c|c}
 \input{figures/fig12atsp7.pspdftex}& \input{figures/fig12atsp8.pspdftex}\\
 $(a)$ & $(b)$
 \end{tabular}
\end{center}
\caption{2. Case $\phi(x_i) \oplus \phi(x_{i+1})=0$, 
$\phi(x_i) \oplus \phi(x_{j})=1$
and $\phi(x_j) \oplus \phi(x_{j+1})=0$. }
\label{fig:12atsp:6}
\end{figure} 
\subsubsection{Traversing Graphs Corresponding to  Matching Equations}
Let  $\cC_l$ be a circle in $\cH$ and $x^l_i \oplus x^l_{j}=0$ with $e=\{i,j\}\in M_l$ 
a matching equation.
Given  $x_i \oplus x_{i+1}=0$, $x_i \oplus x_{j}=0$, 
$x_j \oplus x_{j+1}=0$
and the assignment $\phi$, we are going to construct a tour  through the corresponding parity graphs  
in dependence of $\phi$.
 We begin with the case $\phi(x_i) \oplus \phi(x_{i+1})=0$, 
$\phi(x_i) \oplus \phi(x_{j})=0$
and $\phi(x_j) \oplus \phi(x_{j+1})=0$.\\
\\
\noindent
\textbf{1. Case  $\phi(x_i) \oplus \phi(x_{i+1})=0$, 
$\phi(x_i) \oplus \phi(x_{j})=0$
and $\phi(x_j) \oplus \phi(x_{j+1})=0$:}\\
In this case, we traverse the corresponding parity graphs as depicted in 
Figure~\ref{fig:12atsp:5}.
In Figure~\ref{fig:12atsp:5} $(a)$,
we have $\phi(x_i)=\phi(x_{i+1})=\phi(x_j)=\phi(x_{j+1})=1$, whereas in Figure~\ref{fig:12atsp:5} $(b)$, we have   
$\phi(x_i)=\phi(x_{i+1})=\phi(x_j)=\phi(x_{j+1})=0$. 
In both cases, this part of the tour has local length $5$.\\

\textbf{2. Case $\phi(x_i) \oplus \phi(x_{i+1})=0$, 
$\phi(x_i) \oplus \phi(x_{j})=1$
and $\phi(x_j) \oplus \phi(x_{j+1})=0$:}\\
The tour $\sigma_{\phi}$ is pictured in Figure~\ref{fig:12atsp:6} $(a)$ and $(b)$.

In the case $\phi(x_i)=\phi(x_{i+1})=0$ and $\phi(x_j)=\phi(x_{j+1})=1$
depicted in Figure~\ref{fig:12atsp:6} $(a)$, we are forced to enter and 
leave the parity graph $P^l_e$ via $2$-arcs.  So far, we 
associate the local length $6$ with this part of the tour. 

In  Figure~\ref{fig:12atsp:6} $(b)$, we have   
$\phi(x_i)=\phi(x_{i+1})=1$ and $\phi(x_j)=\phi(x_{j+1})=0$. This part of the tour
$\sigma_{\phi}$ contains one $2$-arc yielding the local length $6$.\\
\\
\textbf{3. Case $\phi(x_i) \oplus \phi(x_{i+1})=0$, 
$\phi(x_i) \oplus \phi(x_{j})=0$
and $\phi(x_j) \oplus \phi(x_{j+1})=1$:}\\
In dependence of $\phi$, we traverse the corresponding parity graphs in the way
as depicted in 
Figure~\ref{fig:12atsp:7}. 

\begin{figure}[h]
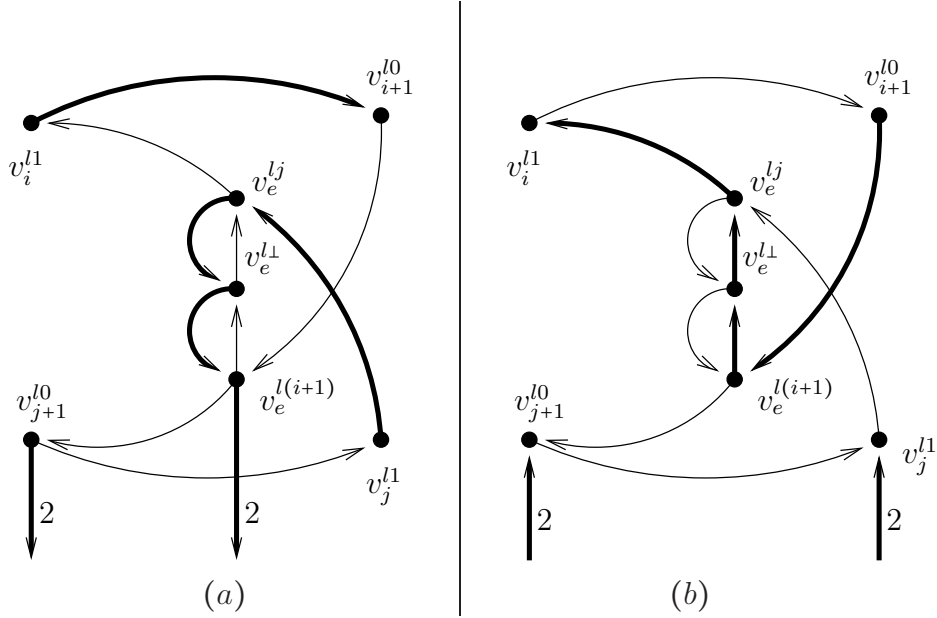

\begin{center}
\begin{tabular}[c]{c|c}
 \input{figures/fig12atsp9.pspdftex}& \input{figures/fig12atsp10.pspdftex}\\
  $(a)$ & $(b)$
 \end{tabular}
\end{center}
\caption{3. Case  $\phi(x_i) \oplus \phi(x_{i+1})=0$, 
$\phi(x_i) \oplus \phi(x_{j})=0$
and $\phi(x_j) \oplus \phi(x_{j+1})=1$. }
\label{fig:12atsp:7}
\end{figure} 
The situation, in which $\phi(x_i)=\phi(x_{i+1})=1$ and $\phi(x_j)=1\neq  \phi(x_{j+1})$ holds, is depicted
in  Figure~\ref{fig:12atsp:7} $(a)$. On the other hand, if we have   
$\phi(x_i)=\phi(x_{i+1})=0$ and  $\phi(x_j)=0\neq  \phi(x_{j+1})$, the tour is pictured in Figure~\ref{fig:12atsp:7} $(b)$.
 In both cases, 
we associate the local  length $6$.

\begin{figure}[h!]
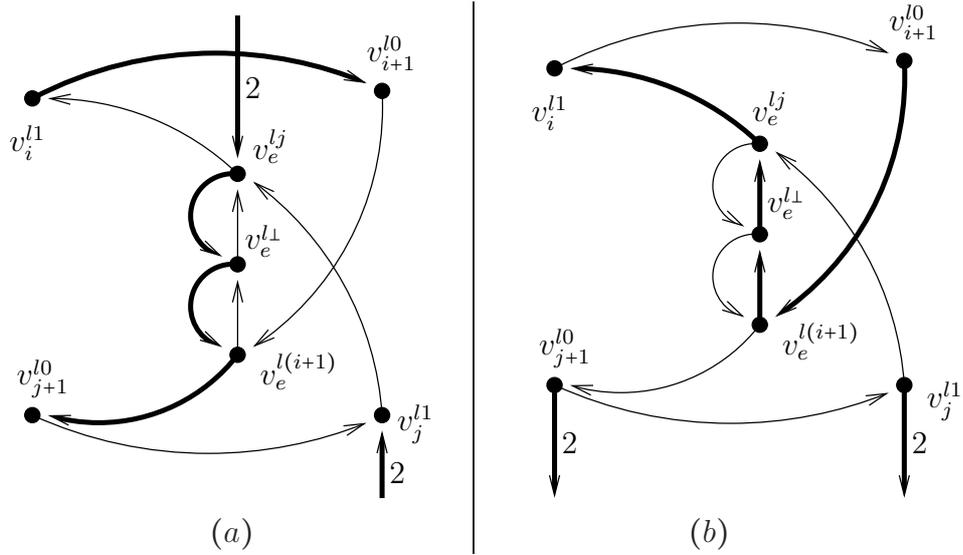

\begin{center}
\begin{tabular}[c]{c|c}
 \input{figures/fig12atsp11.pspdftex}& \input{figures/fig12atsp12.pspdftex}\\
 $(a)$ & $(b)$
 \end{tabular}
\end{center}
\caption{4. Case $\phi(x_i) \oplus \phi(x_{i+1})=0$, 
$\phi(x_i) \oplus \phi(x_{j})=1$
and $\phi(x_j) \oplus \phi(x_{j+1})=1$. }
\label{fig:12atsp:8}
\end{figure}
%
\textbf{4. Case $\phi(x_i) \oplus \phi(x_{i+1})=0$, 
$\phi(x_i) \oplus \phi(x_{j})=1$
and $\phi(x_j) \oplus \phi(x_{j+1})=1$:}\\
The tour $\sigma_{\phi}$ is displayed in  
Figure~\ref{fig:12atsp:8}. 
In  Figure~\ref{fig:12atsp:8} $(a)$,
we are given  $\phi(x_i)=\phi(x_{i+1})=1$ and $\phi(x_j)\neq \phi(x_{j+1})=1$, whereas in 
Figure~\ref{fig:12atsp:8} $(b)$, we have   
$\phi(x_i)=\phi(x_{i+1})=0$ and $\phi(x_j)\neq \phi(x_{j+1})=0$. In both cases, we associate the local
 length $6$ with this 
part of the tour.

\begin{figure}[h!]
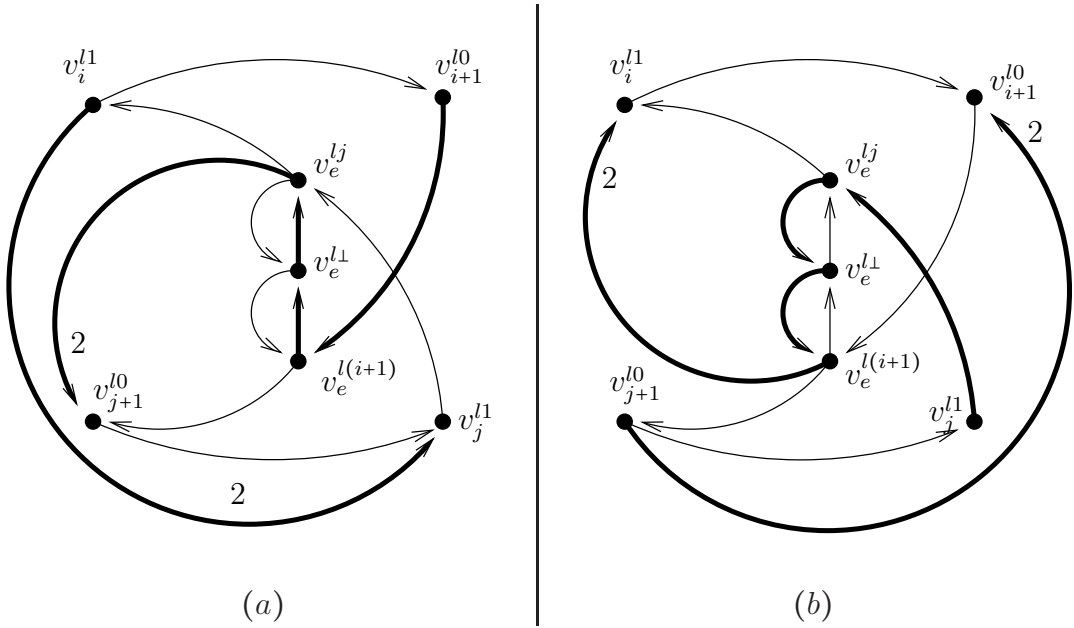

\begin{center}
\begin{tabular}[c]{c|c}
 \input{figures/fig12atsp13.pspdftex}& \input{figures/fig12atsp14.pspdftex}\\
   $(a)$ & $(b)$
 \end{tabular}
\end{center}
\caption{5. Case  $\phi(x_i) \oplus \phi(x_{i+1})=1$, $\phi(x_i) \oplus \phi(x_{j})=1$, 
$\phi(x_j) \oplus \phi(x_{j+1})=1$. }
\label{fig:12atsp:9}
\end{figure} 
\begin{figure}[h!]
\begin{center}
\begin{tabular}[c]{c|c}
 \input{figures/fig12atsp15.pspdftex}& \input{figures/fig12atsp16.pspdftex}\\
  $(a)$ & $(b)$
 \end{tabular}
\end{center}
\caption{6. Case  $\phi(x_i) \oplus \phi(x_{i+1})=1$, $\phi(x_i) \oplus \phi(x_{j})=0$, 
and $\phi(x_j) \oplus \phi(x_{j+1})=1$. }
\label{fig:12atsp:10}
\end{figure} 
\noindent
\textbf{5. Case  $\phi(x_i) \oplus \phi(x_{i+1})=1$, $\phi(x_i) \oplus \phi(x_{j})=1$ and 
$\phi(x_j) \oplus \phi(x_{j+1})=1$:}\\
In this case, we traverse the corresponding parity graphs as depicted in 
Figure~\ref{fig:12atsp:9}.\\
%
In Figure~\ref{fig:12atsp:9} $(a)$,
we notice that $\phi(x_i)\neq \phi(x_{i+1})=0$ and $\phi(x_j) \neq \phi(x_{j+1})=1$, 
whereas in $(b)$, we have   
$\phi(x_i) \neq \phi(x_{i+1})=1$ and  $\phi(x_j)\neq \phi(x_{j+1})=0$. This part of the tour has 
local length $7$.\\
\\
\noindent
\textbf{6. Case  $\phi(x_i) \oplus \phi(x_{i+1})=1$, $\phi(x_i) \oplus \phi(x_{j})=0$ and
$\phi(x_j) \oplus \phi(x_{j+1})=1$:}\\
In this case, we traverse the corresponding parity graphs as depicted in 
Figure~\ref{fig:12atsp:10}. 
In Figure~\ref{fig:12atsp:10} $(a)$,
we have  $\phi(x_i)\neq \phi(x_{i+1})=0$ and $\phi(x_j)\neq \phi(x_{j+1})=0$, whereas in $(b)$, we have   
$\phi(x_i)\neq \phi(x_{i+1})=1$ and  $\phi(x_j)\neq \phi(x_{j+1})=1$. This part of the tour has
local length $7$.

%
%
%
\begin{figure}[h!]
\begin{center}
\begin{tabular}[c]{c|c}
 \input{figures/fig12atsp17.pspdftex}& \input{figures/fig12atsp18.pspdftex}\\
 $(a)$ & $(b)$
 \end{tabular}
\end{center}
\caption{7. Case  $\phi(x_i) \oplus \phi(x_{i+1})=1$, $\phi(x_i) \oplus \phi(x_{j})=0$ and
$\phi(x_j) \oplus \phi(x_{j+1})=0$. }
\label{fig:12atsp:11}
\end{figure} 
\begin{figure}[h!]
\begin{center}
\begin{tabular}[c]{c|c}
 \input{figures/fig12atsp19.pspdftex}& \input{figures/fig12atsp20.pspdftex}\\
   $(a)$ & $(b)$
 \end{tabular}
\end{center}
\caption{8. Case  $\phi(x_i) \oplus \phi(x_{i+1})=1$, $\phi(x_i) \oplus \phi(x_{j})=1$ 
and $\phi(x_j) \oplus \phi(x_{j+1})=0$. }
\label{fig:12atsp:12}
\end{figure} 
\noindent
\textbf{7. Case  $\phi(x_i) \oplus \phi(x_{i+1})=1$, $\phi(x_i) \oplus \phi(x_{j})=0$ and
$\phi(x_j) \oplus \phi(x_{j+1})=0$:}\\
In this case, we traverse the corresponding parity graphs as depicted in 
Figure~\ref{fig:12atsp:11}. 
In  Figure~\ref{fig:12atsp:11} $(a)$,
we have $\phi(x_i) \neq \phi(x_{i+1})=1$ and  $\phi(x_j)=\phi(x_{j+1})=0$, whereas in $(b)$, we have   
$\phi(x_i) \neq \phi(x_{i+1})=0$ and $\phi(x_j)=\phi(x_{j+1})=1$. This part of the tour has
local  length $6$.\\
\\
\textbf{8. Case  $\phi(x_i) \oplus \phi(x_{i+1})=1$, $\phi(x_i) \oplus \phi(x_{j})=1$ 
and $\phi(x_j) \oplus \phi(x_{j+1})=0$:}\\
In the final case, we traverse the corresponding parity graphs as depicted in 
Figure~\ref{fig:12atsp:12}. 
In  Figure~\ref{fig:12atsp:12} $(a)$,
we notice that  $\phi(x_i)\neq \phi(x_{i+1})=0$ and $\phi(x_j)=\phi(x_{j+1})=0$, 
whereas in $(b)$, we have   
$\phi(x_i) \neq \phi(x_{i+1})=1$ and $\phi(x_j)=\phi(x_{j+1})=1$. 
This part of the tour has local length $6$.\\
\\
Our analysis yields the following proposition.
\begin{proposition}\label{prep:phitosigmat}
Let   $x^l_i \oplus x^l_{i+1}=0$, $x^l_i \oplus x^l_{j}=0$ 
and
$x^l_j \oplus x^l_{j+1}=0$ be  equations in $\cH$. Given an 
assignment $\phi$ to the variables in $\cH$,
the associated  tour $\sigma_{\phi}$ has local length at most
$5 +u$,
where $u$ denotes the number of unsatisfied equations in $\{x^l_i \oplus x^l_{i+1}=0, 
x^l_i \oplus x^l_{j}=0, x^l_j \oplus x^l_{j+1}=0\}$ by $\phi$.
\end{proposition}   
\subsubsection{Traversing Graphs Corresponding to  Equations with Three Variables} 
Let $g^3_c\equiv x^l_i \oplus x^s_j \oplus x^k_t =0$ be an equation with three variables in $\cH$.
Furthermore, let $x^l_i \oplus x^l_{i+1}=0$, $x^s_j \oplus x^s_{j+1}=0$ and $x^k_t \oplus x^k_{t+1}=0$ be 
circle equations in $\cH$. 
For notational simplicity, we introduce $e=\{i,i+1\}$, $a=\{t,t+1\}$ and $b=\{j,j+1\}$. 
In Figure~\ref{fig:12atspa1}, we display 
the construction involving the graphs $D^3_c$, $P^k_a$, $P^s_b$, $P^l_e$, $P^l_i$ and $P^l_{i+1}$. 
Exemplary, we depicted  the connections of the graphs $D^3_c$, $P^l_e$, $P^l_i$ and $P^l_{i+1}$  in this
figure.
We are going to construct
the tour $\sigma_{\phi}$ traversing the corresponding graphs and 
analyze the dependency of the local length of $\sigma_{\phi}$
and the number of satisfied participating equations.

\begin{figure}[h!]
\begin{center}
\input{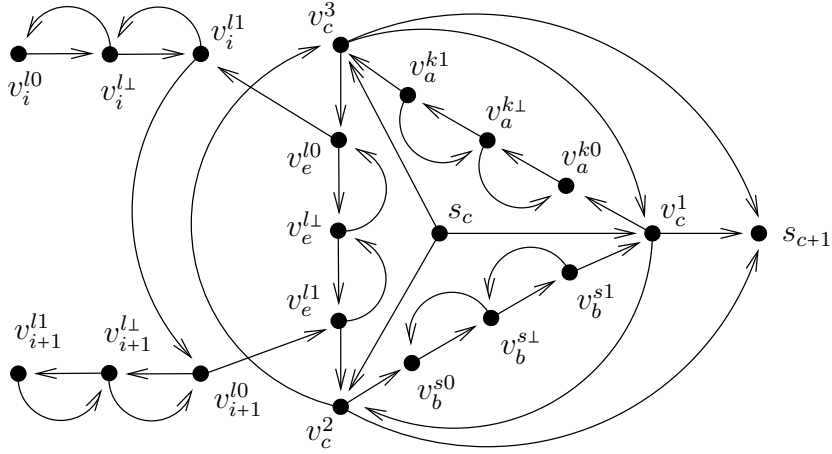}
 \end{center}
\caption{The graph $D^3_c$ with its connections to $P^l_i$ and $P^l_{i+1}$.}
\label{fig:12atspa1}
\end{figure}  

Recall from Proposition~\ref{pro:gadget3atsp} that 
there is a Hamiltonian path from $s_c$ to  $s_{c+1}$ containing the vertices $v^1_c$,
$v^2_c$ and $v^3_c$ in $D^3_{c}$ if and only if an even number of 
parity graphs  $P\in \{P^k_a, P^s_b, P^l_e\}$ is traversed.\\
The outer loop traverses the graph $D^3_c$ starting at $s_c$ and ending at $s_{c+1}$. Furthermore, it
contains the vertices $v^1_c$, $v^2_c$ and $v^3_c$ in some order. If $\sigma_{\phi}$ 
traverses an even number of parity graphs $P\in \{P^k_a, P^s_b, P^l_e\}$ in the inner loop,
it is possible to construct a Hamiltonian path with associated 
local length $3\cdot3+4$. In the other case, we have to use a $2$-arc yielding the 
local length $14$. 

\begin{figure}[h!]
\begin{center}
\input{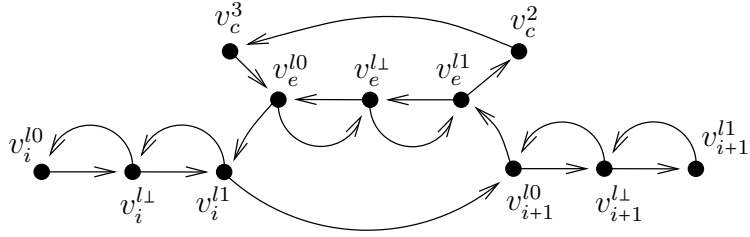}
 \end{center}
\caption{A part of the graphs corresponding to $x^l_{i}\oplus x^l_{i+1}=0$
and $x^l_i \oplus x^s_j \oplus x^k_t =0$. }
\label{fig:12atspa2}
\end{figure} 

Let us analyze the part of $\sigma_{\phi}$ traversing graphs corresponding to 
$x^l_{i}\oplus x^l_{i+1}=0$.   
For this reason, we will examine the situation depicted in Figure~\ref{fig:12atspa2}. 
%
%
%
%
Let us 
begin with the case
$\phi(x^l_{i})\oplus \phi(x^l_{i+1}) =0$.\\
\\
\noindent
\textbf{1. Case $\phi(x^l_{i})\oplus \phi(x^l_{i+1}) =0$ :} \\
If $\phi(x^l_{i})= \phi(x^l_{i+1})=1$ holds, the tour $\sigma_{\phi}$ uses
the arc $(v^{l1}_{i},v^{l0}_{i+1})$. Afterwards, the parity graph $P^l_{e}$
will be traversed when the tour leads through the graph $D^3_c$.
More precisely, it will use the path 
$v^3_c \rightarrow  v^{l0}_e \rightarrow v^{l\bot}_e 	\rightarrow v^{l1}_e 	\rightarrow v^{2}_c$.
In Figure~\ref{fig:12atspa3}, we illustrated this part of the tour.

\begin{figure}[h!]
\begin{center}
\input{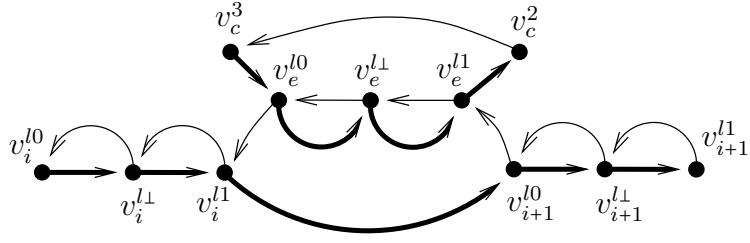}
 \end{center}
\caption{Case $\phi(x^l_{i})= \phi(x^l_{i+1})=1$.}
\label{fig:12atspa3}
\end{figure}

In the other case $\phi(x^l_{i})= \phi(x^l_{i+1})=0$, we use the path  
$v^{l0}_{i+1} \rightarrow  v^{l1}_e \rightarrow v^{l\bot}_e 	\rightarrow v^{l0}_e 	\rightarrow v^{l1}_{i}$.
Afterwards, the tour $\sigma_{\phi}$ contains the arc $(v^2_c,v^3_c)$. 

In both cases, we associate the local length $1$ with this part of the tour. \\
\\ 
\textbf{2. Case $\phi(x^l_{i+1})\oplus \phi(x^l_{i+1}) =1$ :} \\
Assuming $\phi(x^l_{i}) \neq \phi(x^l_{i+1})=1$, the tour $\sigma_{\phi}$
uses a $2$-arc entering $v^{l1}_e$ and the path 
$v^{l1}_e \rightarrow v^{l\bot}_e 	\rightarrow v^{l0}_e 	\rightarrow v^{l1}_{i}$.
Furthermore, we need another $2$-arc in order to reach $v^{l0}_{i+1}$. The situation
is depicted in Figure~\ref{fig:12atspa4}.

\begin{figure}[h]
\begin{center}
\input{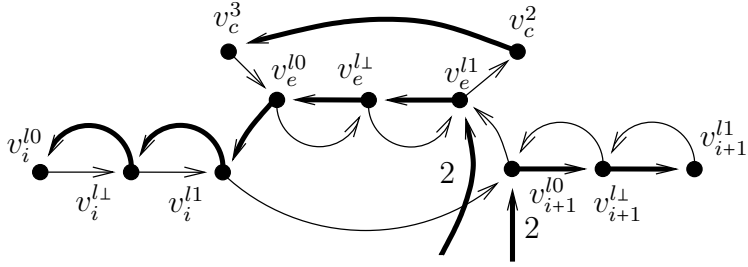}
 \end{center}
\caption{Case $\phi(x^l_{i}) \neq \phi(x^l_{i+1})=1$.}
\label{fig:12atspa4}
\end{figure}

\noindent
In the other case, namely $\phi(x^l_{i}) \neq \phi(x^l_{i+1})=0$, we use $2$-arcs leaving
$v^{l0}_{i+1}$ and $v^{l1}_{i}$. Afterwards, the tour uses the path
$v^3_c \rightarrow  v^{l0}_e \rightarrow v^{l\bot}_e 	\rightarrow v^{l1}_e 	\rightarrow v^{2}_c$
while traversing the graph $D^3_c$.

In both cases, we associate the local length $2$ with this part of the tour.\\
\\
We obtain the following proposition.
\begin{proposition}\label{prep:phitosigeq3}
Let  $x^l_i \oplus x^s_j \oplus x^k_t =0$ be an equation with three variables in $\cH$.
Furthermore, let $x^l_i \oplus x^l_{i+1}=0$, $x^s_j \oplus x^s_{j+1}=0$ and
$x^k_t \oplus x^k_{t+1}=0$ be circle  equations in $\cH$. Given an 
assignment $\phi$ to the variables in $\cH$,
the associated  tour $\sigma_{\phi}$ has local length  at most
$3\cdot3+4 +3+u$, 
where $u$ denotes the number of unsatisfied  equations in 
$\{x^l_i \oplus x^s_j \oplus x^k_t =0,x^l_i \oplus x^l_{i+1}=0,
x^s_j \oplus x^s_{j+1}=0,x^k_t \oplus x^k_{t+1}=0\}$ by $\phi$.
\end{proposition} 

\subsubsection{Traversing Graphs Corresponding to  Circle Border Equations}
Let $\cC_l$ be a circle in $\cH$ and $x^l_1 \oplus x^l_n=0$ its circle border equation.
Recall that the variable $x^l_n$ is also included in an equation with three
variables. We are going to describe the part of the tour passing through the graphs
 depicted in Figure~\ref{fig:12atspa5} in dependence of the assigned values to the variables $x^l_1$
 and  $x^l_n$. Let us start with the case $\phi(x^l_1) \oplus \phi(x^l_n) =0$.

\begin{figure}[h!]
\begin{center}
\input{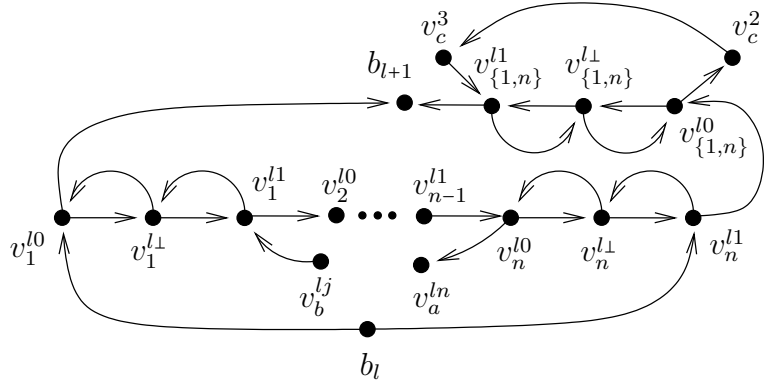}
 \end{center}
\caption{Traversing Graphs Corresponding To  Circle Border Equations.}
\label{fig:12atspa5}
\end{figure}

\noindent 
\textbf{1. Case $\phi(x^l_1) \oplus \phi(x^l_n) =0$ }\\
The starting point of the tour $\sigma_{\phi}$ passing through the graph corresponding to $x^l_1 \oplus x^l_n=0$ 
is the vertex $b_l$. 
Given the values $\phi(x^l_1)=\phi(x^l_n)$, we use in each case the $\phi(x^l_1)$-traversal of  
the parity graphs $P^l_1$ and $P^l_n$ ending in $b^l_{l+1}$. Note that in the case 
$\phi(x^l_1)=\phi(x^l_n)=0$, we use the $1$-traversal of the parity graph $P^l_{\{1,n\}}$.  
Exemplary, we display the 
situation $\phi(x^l_1)=\phi(x^l_n)=1$ in Figure~\ref{fig:12atspa6}. 

\begin{figure}[h]
\begin{center}
\input{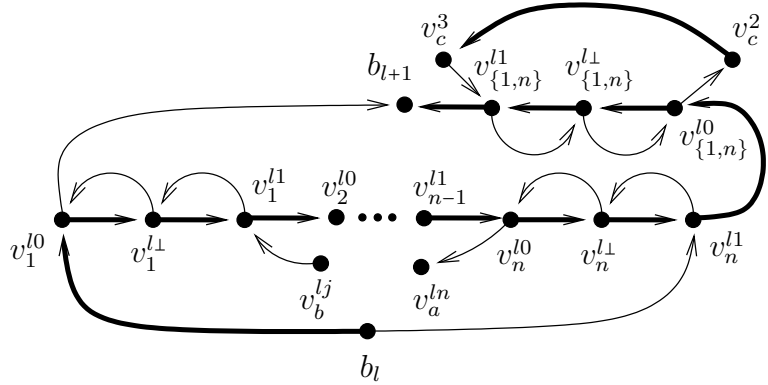}
 \end{center}
\caption{Case $\phi(x^l_1)=\phi(x^l_n)=1$.}
\label{fig:12atspa6}
\end{figure}

In both cases, we associated the local
length $2$ with this part of the tour.\\
\\  
 \textbf{2. Case $\phi(x^l_1) \oplus \phi(x^l_n) =1$ }\\
Given the assignment $\phi(x^l_1)\neq \phi(x^l_n)=0$, we traverse the arc $(b_l,v^{l0}_1)$.
Due to the construction, we have to use  $2$-arcs to enter $b_{l+1}$ and $v^{l1}_n$
as depicted in Figure~\ref{fig:12atspa7}.

\begin{figure}[h]
\begin{center}
\input{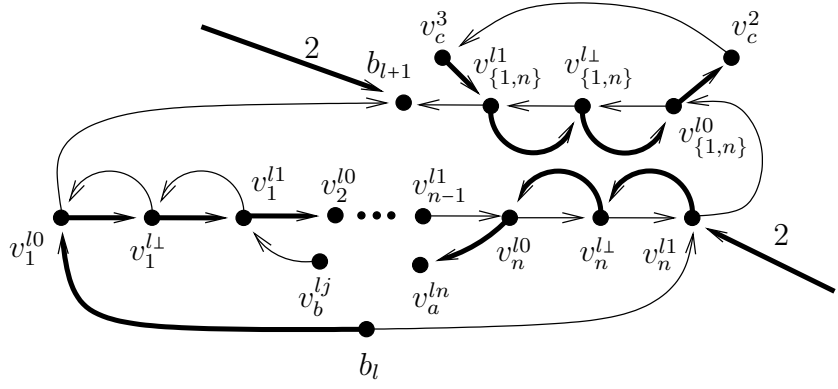}
 \end{center}
\caption{Case $\phi(x^l_1)\neq \phi(x^l_n)=0$.}
\label{fig:12atspa7}
\end{figure}

In the other case, we have to use a $2$-arc in order to leave the vertex $b_l$. In addition,
the tour contains the path  $v^{l1}_n \rightarrow v^{l0}_{\{1,n\}}   
 \rightarrow v^{l\bot }_{\{1,n\}}\rightarrow v^{l1}_{\{1,n\}}\rightarrow b_{l+1}$.
 
Hence, in both cases, we associate the local length $3$ with this part of $\sigma_{\phi}$.\\
\\
We obtain the following proposition.
\begin{proposition}\label{prep:phitosigcirlb}
Let  $x^l_1 \oplus x^l_{n}=0$ be an  circle border equations in $\cH$. Given an 
assignment $\phi$ to the variables in $\cH$,
the associated  tour $\sigma_{\phi}$ has local length  
$2$ if the equation  $x^l_1 \oplus x^l_{n}=0$ is satisfied 
 by $\phi$ and $3$, otherwise.
\end{proposition}  
%
%

\subsection{ Constructing the Assignment $\psi_{\sigma}$ from a Tour $\sigma$ }
Let $\cH$ be an instance of the Hybrid problem, $D_{\cH}=(V_{\cH}, A_{\cH})$ the 
associated instance of the $(1,2)$-ATSP problem and $\sigma$ a tour in $D_{\cH}$.
We are going to define the corresponding assignment $\psi_{\sigma}:V(\cH) \rightarrow \{0,1\}$
to the variables in $\cH$. In addition to it, we establish a connection between 
the length of $\sigma$ and the number of satisfied equations by $\psi_{\sigma}$.\\
\\ 
In order to define an assignment, we first introduce the notion of
consistent tours in $D_\cH$.  
\begin{definition}[Consistent Tour]
Let $\cH$ be an instance of the Hybrid problem and  $D_{\cH}$ the associated instance of the
$(1,2)$--ATSP problem.
A tour in $D_{\cH}$ is called consistent if the tour uses only $0/1$-traversals of 
all in $D_{\cH}$ contained parity graphs. 
\end{definition}
\noindent
Due to the following proposition, we may assume that the underlying tour is consistent. 
\begin{proposition}\label{prop:consistent}
Let $\cH$ be an instance of the Hybrid problem and $D_{\cH}$ the associated instance of the
$(1,2)$--ATSP problem. Any tour $\sigma$ in $D_{\cH}$ can be transformed 
in polynomial time into a consistent tour $\sigma'$ with $\ell(\sigma') \leq \ell(\sigma)$.
\end{proposition}
\begin{proof}
For every parity graph contained in $D_{\cH}$, 
it can be seen by considering all possibilities exhaustively that
any tour in $D_{\cH}$ that is not using the corresponding $0/1$-traversals can be modified into a
tour with at most the same number of $2$-arcs. The less
obvious cases are shown in Figure~\ref{figcases12atsp}~(see 10. Figure Appendix).
\end{proof}
\noindent
In the remainder, we assume that the underlying tour $\sigma$ is consistent with all parity graphs
in $D_{\cH}$.
Let us now define the corresponding assignment $\psi_{\sigma}$ given  $\sigma$.
\begin{definition}[ Assignment $\psi_{\sigma}$]\label{def:psi}
Let $\cH$ be an instance of the Hybrid problem, $D_{\cH}=(V_{\cH}, A_{\cH})$ the 
associated instance of the $(1,2)$-ATSP problem. 
Given a tour $\sigma$ in $D_{\cH}$, in which all parity graphs are consistent with respect to
$\sigma$, the corresponding  assignment $\psi_{\sigma}:V(\cH) \rightarrow \{0,1\}$ is defined as follows.
\begin{eqnarray*}
\psi_{\sigma}(x^l_i) & =1 & \textrm{ if $\sigma$ uses a $1$-traversal of $P^l_{i}$  } \\
                     &   = 0 &  \textrm{ otherwise  }
\end{eqnarray*}
\end{definition}
%
%
%
We are going to analyze the local length of $\sigma$ in dependence of  
the number of corresponding satisfied equations by $\psi_{\sigma}$.
In some cases, we will have to modify the underlying tour improving 
in this way on the number of satisfied equations by 
the corresponding assignment  $\psi_{\sigma}$. Let us start with the analysis.
\subsubsection{Transforming $\sigma$ in Graphs Corresponding to  Matching Equations}  
Given the equations $x_i \oplus x_{i+1}=0$, $x_i \oplus x_{j}=0$, 
$x_j \oplus x_{j+1}=0$
and a tour $\sigma$, we are going to construct an assignment in dependence of $\sigma$.
In particular, we analyze the relation between the length of the tour and
the number of satisfied equations by $\psi_{\sigma}$.\\
\\
\begin{figure}[h!]
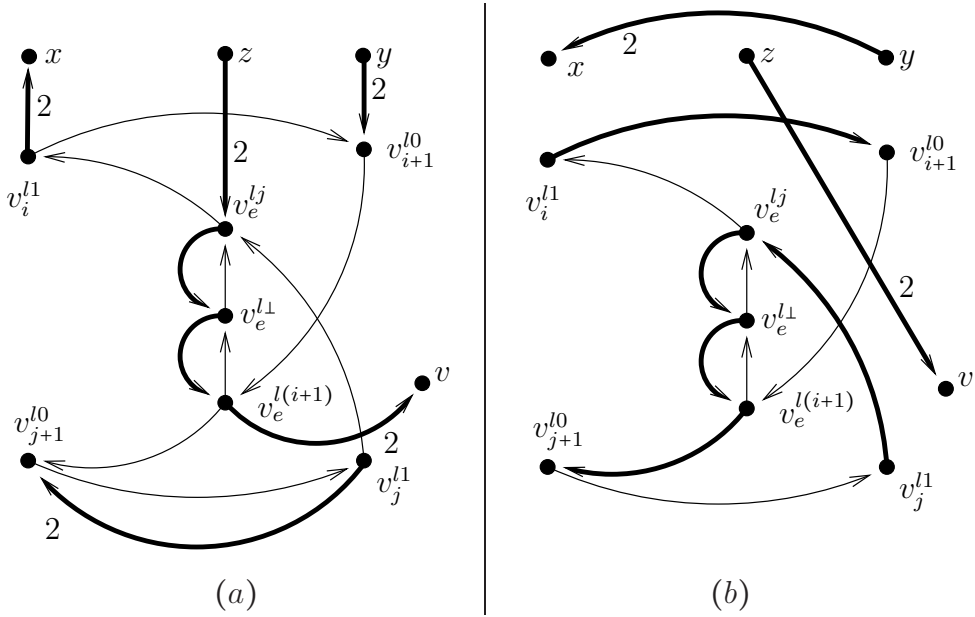

\begin{center}
\begin{tabular}[c]{c|c}
 \input{figures/fig12atsp22.pspdftex}& \input{figures/fig12atsp23.pspdftex}\\
 $(a)$ & $(b)$
 \end{tabular}
\end{center}
\caption{1.Case  $\psi_{\sigma}(x_i) \oplus \psi_{\sigma}(x_{i+1})=0$, 
$\psi_{\sigma}(x_i) \oplus \psi_{\sigma}(x_{j})=0$ \&
$\psi_{\sigma}(x_j) \oplus \psi_{\sigma}(x_{j+1})=0$.}
\label{fig:12atsp:13}
\end{figure} 
\noindent
\textbf{1. Case  $\psi_{\sigma}(x_i) \oplus \psi_{\sigma}(x_{i+1})=0$, 
$\psi_{\sigma}(x_i) \oplus \psi_{\sigma}(x_{j})=0$ and 
$\psi_{\sigma}(x_j) \oplus \psi_{\sigma}(x_{j+1})=0$: }\\
Given $\psi_{\sigma}(x_i)=\psi_{\sigma}(x_{j})=\psi_{\sigma}(x_j)=\psi_{\sigma}(x_{j+1})=1$,
 it is possible to transform the underlying tour   such that
no $2$-arcs 
enter or leave the vertices $v^{l1}_{i}$, $v^{l0}_{i+1}$, $v^{l0}_{j+1}$, $v^{lj}_{e}$, $v^{l(i+1)}_{e}$   and $v^{l1}_{j}$.
Exemplary, we display in Figure~\ref{fig:12atsp:13} such an transformation, where Figure~\ref{fig:12atsp:13} $(a)$
and Figure~\ref{fig:12atsp:13} $(b)$
illustrate the underlying tour $\sigma$ and   the transformed tour $\sigma'$, respectively.
The case $\psi_{\sigma}(x_i)=\psi_{\sigma}(x_{i+1})=\psi_{\sigma}(x_j)=\psi_{\sigma}(x_{j+1})=0$
can be discussed analogously.\\
\begin{figure}[h!]
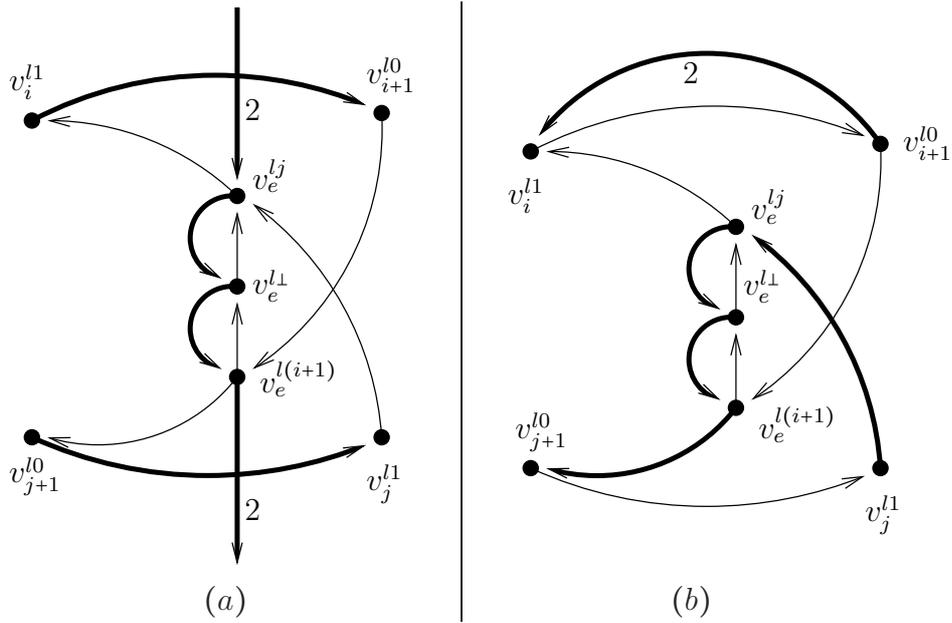

\begin{center}
\begin{tabular}[c]{c|c}
 \input{figures/fig12atsp7.pspdftex}& \input{figures/fig12atsp8.pspdftex}\\
 $(a)$ & $(b)$
 \end{tabular}
\end{center}
\caption{2.Case  $\psi_{\sigma}(x_i) \oplus \psi_{\sigma}(x_{i+1})=0$, 
$\psi_{\sigma}(x_i) \oplus \psi_{\sigma}(x_{j})=1$ \& 
$\psi_{\sigma}(x_j) \oplus \psi_{\sigma}(x_{j+1})=0$.}
\label{fig:12atsp:14}
\end{figure}

In both cases, we obtain the local length $5$ for this part of $\sigma$ while $\psi_{\sigma}$ satisfies all
$3$ equations.\\
\\

\begin{figure}[h!]
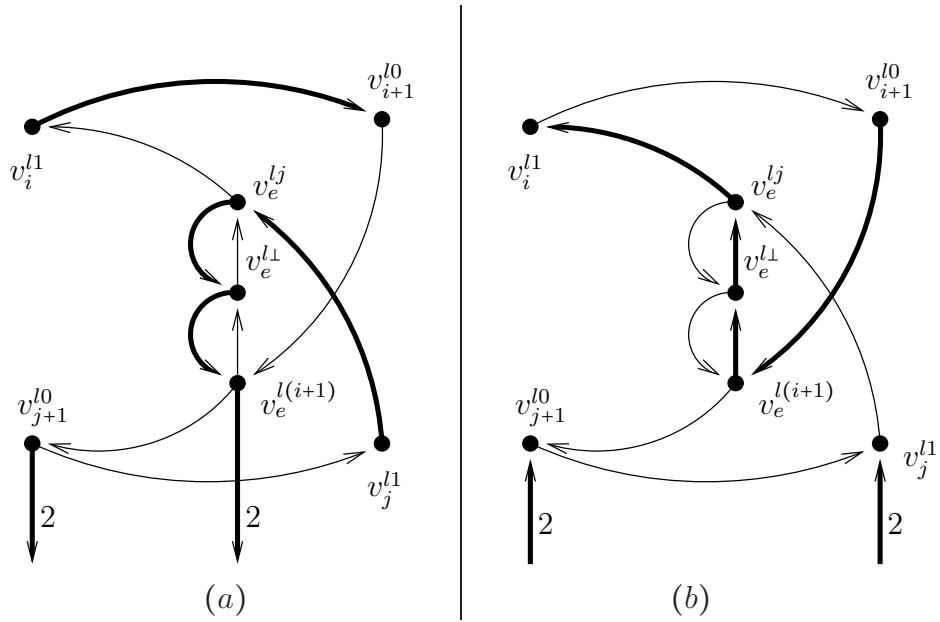

\begin{center}
\begin{tabular}[c]{c|c}
 \input{figures/fig12atsp9.pspdftex}& \input{figures/fig12atsp10.pspdftex}\\
  $(a)$ & $(b)$
 \end{tabular}
\end{center}
\caption{3.Case $\psi_{\sigma}(x_i) \oplus \psi_{\sigma}(x_{i+1})=0$, 
$\psi_{\sigma}(x_i) \oplus \psi_{\sigma}(x_{j})=0$ \& 
$\psi_{\sigma}(x_j) \oplus \psi_{\sigma}(x_{j+1})=1$. }
\label{fig:12atsp:15}
\end{figure}
\begin{figure}[h!]
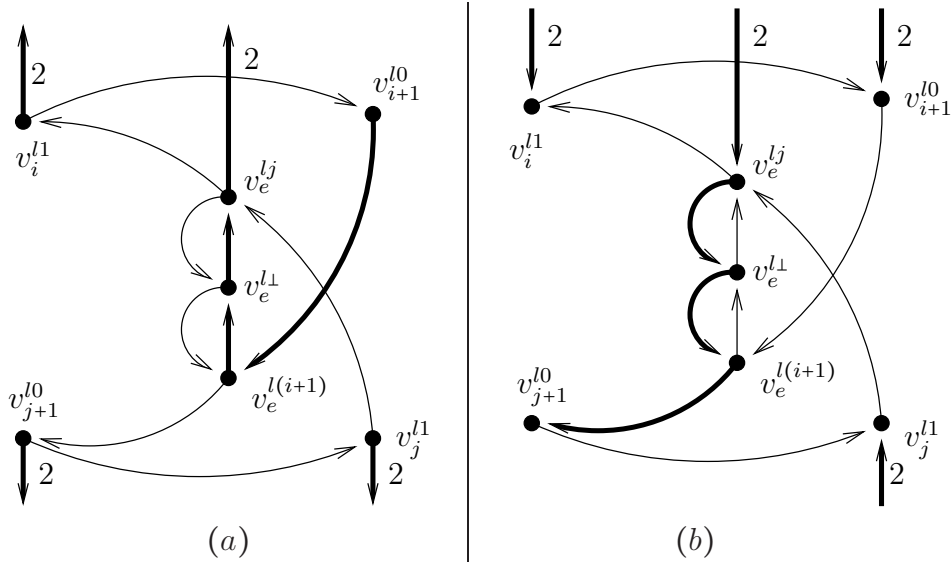

\begin{center}
\begin{tabular}[c]{c|c}
 \input{figures/fig12atsp15.pspdftex}& \input{figures/fig12atsp16.pspdftex}\\
  $(a)$ & $(b)$
 \end{tabular}
\end{center}
\caption{4. Case $\psi_{\sigma}(x_i) \oplus \psi_{\sigma}(x_{i+1})=1$, 
$\psi_{\sigma}(x_i) \oplus \psi_{\sigma}(x_{j})=0$ \&
$\psi_{\sigma}(x_j) \oplus \psi_{\sigma}(x_{j+1})=1$.}
\label{fig:12atsp:20}
\end{figure}

\noindent
\textbf{2. Case  $\psi_{\sigma}(x_i) \oplus \psi_{\sigma}(x_{i+1})=0$, 
$\psi_{\sigma}(x_i) \oplus \psi_{\sigma}(x_{j})=1$ and 
$\psi_{\sigma}(x_j) \oplus \psi_{\sigma}(x_{j+1})=0$: }\\
Assuming $\psi_{\sigma}(x_i)=\psi_{\sigma}(x_{i+1})=1$ and $\psi_{\sigma}(x_j)=\psi_{\sigma}(x_{j+1})=0$,
we are able to transform the tour such that it uses the arcs $(v^{l1}_{i},v^{l0}_{i+1})$ and $(v^{l0}_{j+1},v^{l1}_{j})$. 
Due to the construction and our assumption, the tour cannot traverse the arcs $(v^{l1}_{j},v^{lj}_{e})$,
$(v^{lj}_{e},v^{l1}_{i})$, $(v^{l(i+1)}_{e},v^{l0}_{j+1})$
and $(v^{l0}_{i+1},v^{l(i+1)}_{e})$. Consequently, we have to use  $2$-arcs entering and leaving the parity graph $P^l_e$. 
The situation is displayed in Figure~\ref{fig:12atsp:14} $(a)$. We associate only the cost of one $2$-arc yielding
 the local length $6$, which corresponds to the fact that $\psi_{\sigma}$ leaves the equation $x_i \oplus x_j =0$
 unsatisfied.
  
Note that a similar situation holds in case of $\psi_{\sigma}(x_i)=\psi_{\sigma}(x_{i+1})=0$ 
and $\psi_{\sigma}(x_j)=\psi_{\sigma}(x_{j+1})=1$ (cf. Figure~\ref{fig:12atsp:14} $(b)$).\\
\\
\noindent
\textbf{3. Case  $\psi_{\sigma}(x_i) \oplus \psi_{\sigma}(x_{i+1})=0$, 
$\psi_{\sigma}(x_i) \oplus \psi_{\sigma}(x_{j})=0$ and 
$\psi_{\sigma}(x_j) \oplus \psi_{\sigma}(x_{j+1})=1$: }\\
Let us start with the case
$\psi_{\sigma}(x_i)=\psi_{\sigma}(x_{i+1})=1$ and $\psi_{\sigma}(x_j) \neq \psi_{\sigma}(x_{j+1})=0$.
The situation is displayed in Figure~\ref{fig:12atsp:15} $(a)$. 
Due to the construction, we are able to transform $\sigma$ such that it uses the arc $(v^{l1}_{i},v^{l0}_{i+1})$.
Note that the tour cannot traverse the arcs $(v^{l(i+1)}_{e},v^{l0}_{j+1})$ and  $(v^{l0}_{j+1},v^{l1}_{j})$.
Hence, we are forced to use two $2$-arcs increasing the cost by $2$. 
All in all, we obtain the local length $6$.\\
The case $\psi_{\sigma}(x_i)=\psi_{\sigma}(x_{i+1})=0$ and $\psi_{\sigma}(x_j) \neq \psi_{\sigma}(x_{j+1})=1$
can be analyzed analogously (cf. Figure~\ref{fig:12atsp:15} $(b)$).
A similar argumentation holds for 
$\psi_{\sigma}(x_i) \oplus \psi_{\sigma}(x_{i+1})=1$, 
$\psi_{\sigma}(x_i) \oplus \psi_{\sigma}(x_{j})=0$ and 
$\psi_{\sigma}(x_j) \oplus \psi_{\sigma}(x_{j+1})=0$.\\
\\
%
%
\noindent
\textbf{4. Case  $\psi_{\sigma}(x_i) \oplus \psi_{\sigma}(x_{i+1})=1$, 
$\psi_{\sigma}(x_i) \oplus \psi_{\sigma}(x_{j})=0$ and 
$\psi_{\sigma}(x_j) \oplus \psi_{\sigma}(x_{j+1})=1$: }\\
Given $\psi_{\sigma}(x_i)\neq \psi_{\sigma}(x_{i+1})=0$ and 
$\psi_{\sigma}(x_j) \neq \psi_{\sigma}(x_{j+1})=0$, we are able to
transform the tour such that it uses the arc $(v^{l0}_{i+1},v^{l(i+1)}_{e})$.
This situation is depicted in Figure~\ref{fig:12atsp:20} $(a)$.
Notice that we are forced to use four $2$-arcs in order to connect all vertices.
Consequently, it yields the local length $7$.\\ 
\\
\begin{figure}[h!]
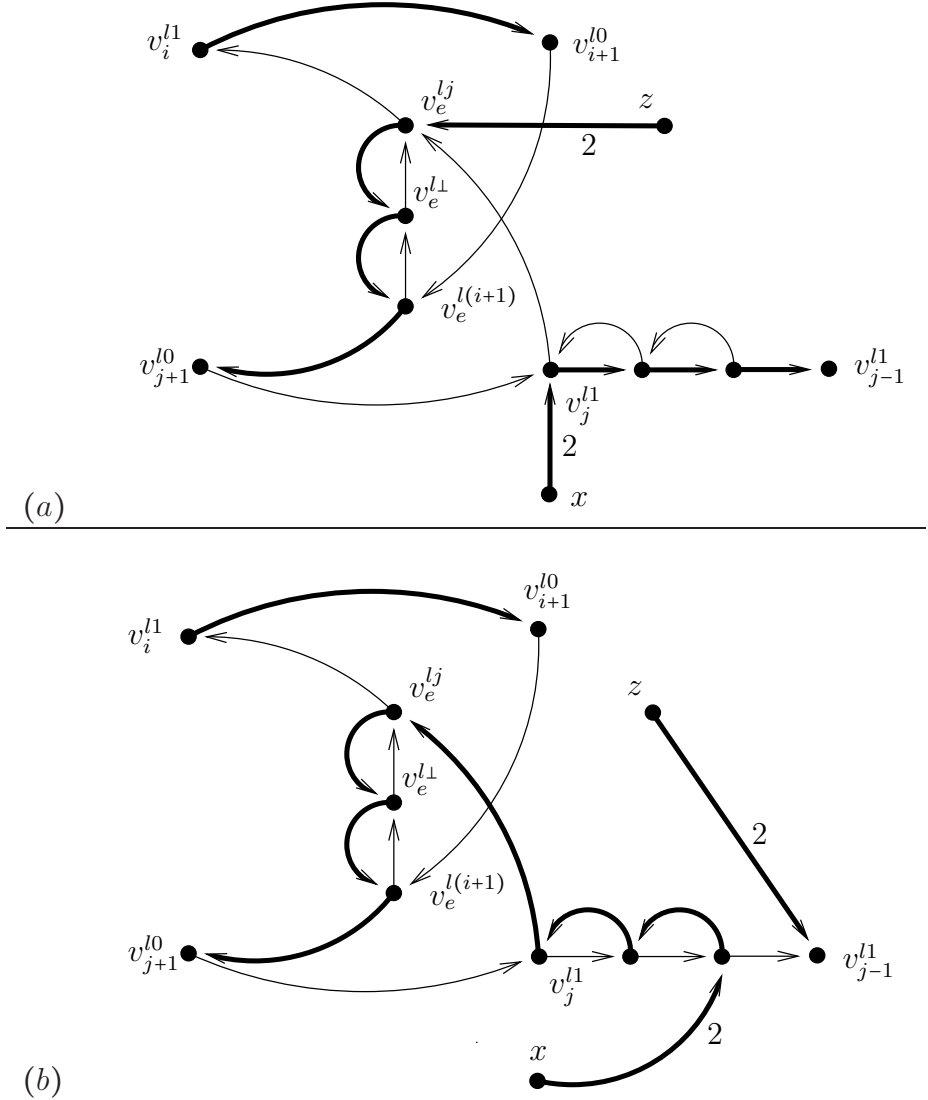

\begin{center}
\begin{tabular}[c]{lc}
 $(a)$ & \input{figures/fig12atsp24.pspdftex} \\
 \hline
 $(b)$ &  \input{figures/fig12atsp25.pspdftex}
 \end{tabular}
\end{center}
\caption{5.Case with $\psi_{\sigma}(x_i)=\psi_{\sigma}(x_{i+1})=1$ and 
$\psi_{\sigma}(x_j) \neq \psi_{\sigma}(x_{j+1})=1$. }
\label{fig:12atsp:16}
\end{figure}
The case, in which $\psi_{\sigma}(x_i)\neq \psi_{\sigma}(x_{i+1})=0$ and 
$\psi_{\sigma}(x_j) \neq \psi_{\sigma}(x_{j+1})=0$ holds, is displayed in 
Figure~\ref{fig:12atsp:20} $(b)$ and can be discussed
analogously.\\
\\
%

\noindent
\textbf{5. Case  $\psi_{\sigma}(x_i) \oplus \psi_{\sigma}(x_{i+1})=0$, 
$\psi_{\sigma}(x_i) \oplus \psi_{\sigma}(x_{j})=1$ and 
$\psi_{\sigma}(x_j) \oplus \psi_{\sigma}(x_{j+1})=1$: }\\
Let the tour $\sigma$ be characterized by 
$\psi_{\sigma}(x_i)=\psi_{\sigma}(x_{i+1})=1$ and 
$\psi_{\sigma}(x_j) \neq \psi_{\sigma}(x_{j+1})=1$.
Then, we  transform $\sigma$ in such a way that we are able to use the arc 
$(v^{l1}_{i},v^{l0}_{i+1})$. The corresponding situation is illustrated 
in Figure~\ref{fig:12atsp:16}~$(a)$. 
In order to change the value of $\psi_{\sigma}(x_j)$, we transform the tour by traversing 
the parity graph $P^l_j$ in the other direction and obtain
$\psi_{\sigma}(x_j)=1$. This transformation induces a tour with at most the same cost. 
On the other hand, the corresponding
assignment $\psi_{\sigma}$ satisfies at least $2-1$ more equations since 
$x^l_j \oplus x^l_{j-1}=0$ might get unsatisfied. In this case, we associate  the local costs
of $6$ with $\sigma$. 
\begin{figure}[h!]
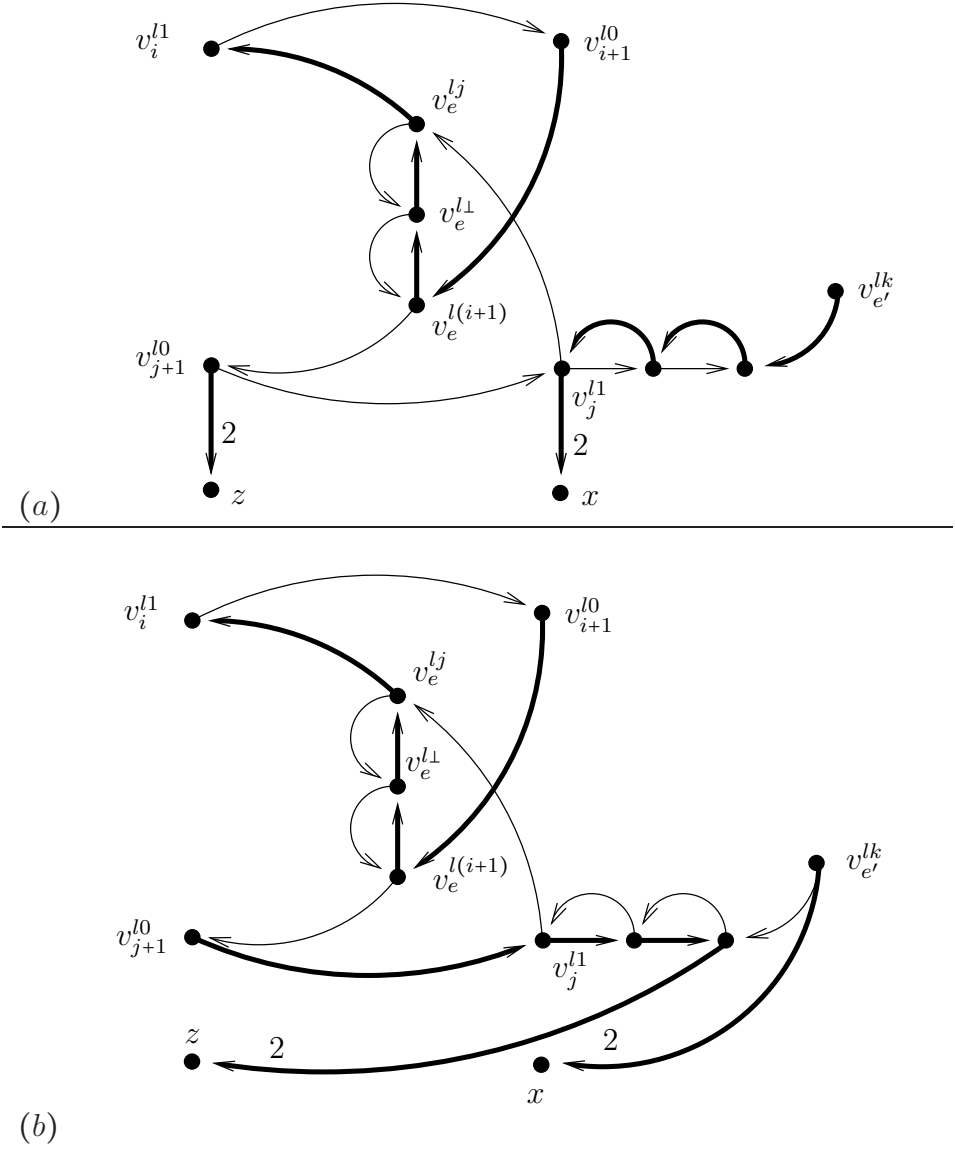

\begin{center}
\begin{tabular}[c]{lc}
 $(a)$ & \input{figures/fig12atsp26.pspdftex} \\
 \hline
 $(b)$ & \input{figures/fig12atsp27.pspdftex}
\end{tabular}
\end{center}
\caption{5.Case with $\psi_{\sigma}(x_i)=\psi_{\sigma}(x_{i+1})=0$ and 
$\psi_{\sigma}(x_j) \neq \psi_{\sigma}(x_{j+1})=0$.  }
\label{fig:12atsp:17}
\end{figure}  

In the other case, in which $\psi_{\sigma}(x_i)=\psi_{\sigma}(x_{i+1})=0$ and 
$\psi_{\sigma}(x_j) \neq \psi_{\sigma}(x_{j+1})=0$ holds, we may argue similarly.
The transformation is depicted in Figure~\ref{fig:12atsp:17} $(a)-(b)$.\\
\\

\noindent
\textbf{6. Case  $\psi_{\sigma}(x_i) \oplus \psi_{\sigma}(x_{i+1})=1$, 
$\psi_{\sigma}(x_i) \oplus \psi_{\sigma}(x_{j})=1$ and 
$\psi_{\sigma}(x_j) \oplus \psi_{\sigma}(x_{j+1})=1$: }\\
Given a tour $\sigma$ with 
$\psi_{\sigma}(x_i)\neq \psi_{\sigma}(x_{i+1})=1$ and 
$\psi_{\sigma}(x_j) \neq \psi_{\sigma}(x_{j+1})=0$,
we transform $\sigma$ such that it traverses the parity graph $P^l_j$
in the opposite direction meaning $\psi_{\sigma}(x_j)=0$ (cf. Figure~\ref{fig:12atsp:18}). 
This transformation enables us to use the arc $(v^{l0}_{j+1},v^{l1}_{j})$.
Furthermore, it yields at least one more satisfied equation in $\cH$.
In order to connect the remaining vertices, we are forced to use at least two $2$-arcs.
In summary, we associate  the local length $7$ with this situation in conformity with
the at most $2$ unsatisfied equations by $\psi_{\sigma}$.\\
\\
 
\begin{figure}[h]
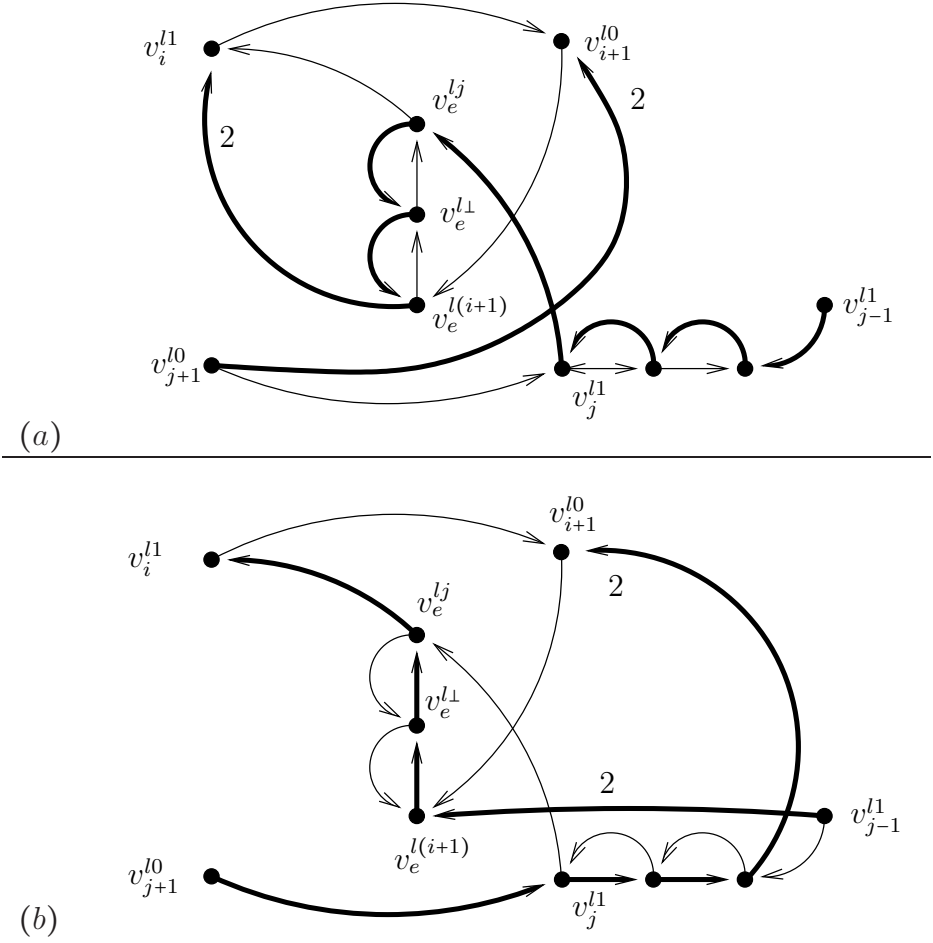

\begin{center}
\begin{tabular}[c]{lc}
 $(a)$ &  \input{figures/fig12atsp30.pspdftex} \\
 \hline
 $(b)$ & \input{figures/fig12atsp31.pspdftex}
\end{tabular}
\end{center}
\caption{6.Case with $\psi_{\sigma}(x_i)\neq \psi_{\sigma}(x_{i+1})=1$ and 
$\psi_{\sigma}(x_j) \neq \psi_{\sigma}(x_{j+1})=0$.   }
\label{fig:12atsp:18}
\end{figure}

If we are given a tour $\sigma$ with  $\psi_{\sigma}(x_i)\neq \psi_{\sigma}(x_{i+1})=0$ and 
$\psi_{\sigma}(x_j) \neq \psi_{\sigma}(x_{j+1})=1$, we obtain the situation displayed
in Figure~\ref{fig:12atsp:19} $(a)$. By applying local transformations without increasing the
length of the underlying tour, we achieve the scenario depicted in Figure~\ref{fig:12atsp:19} $(b)$. 
We argue that the associated local length of the tour is $7$.
 
\begin{figure}[h]
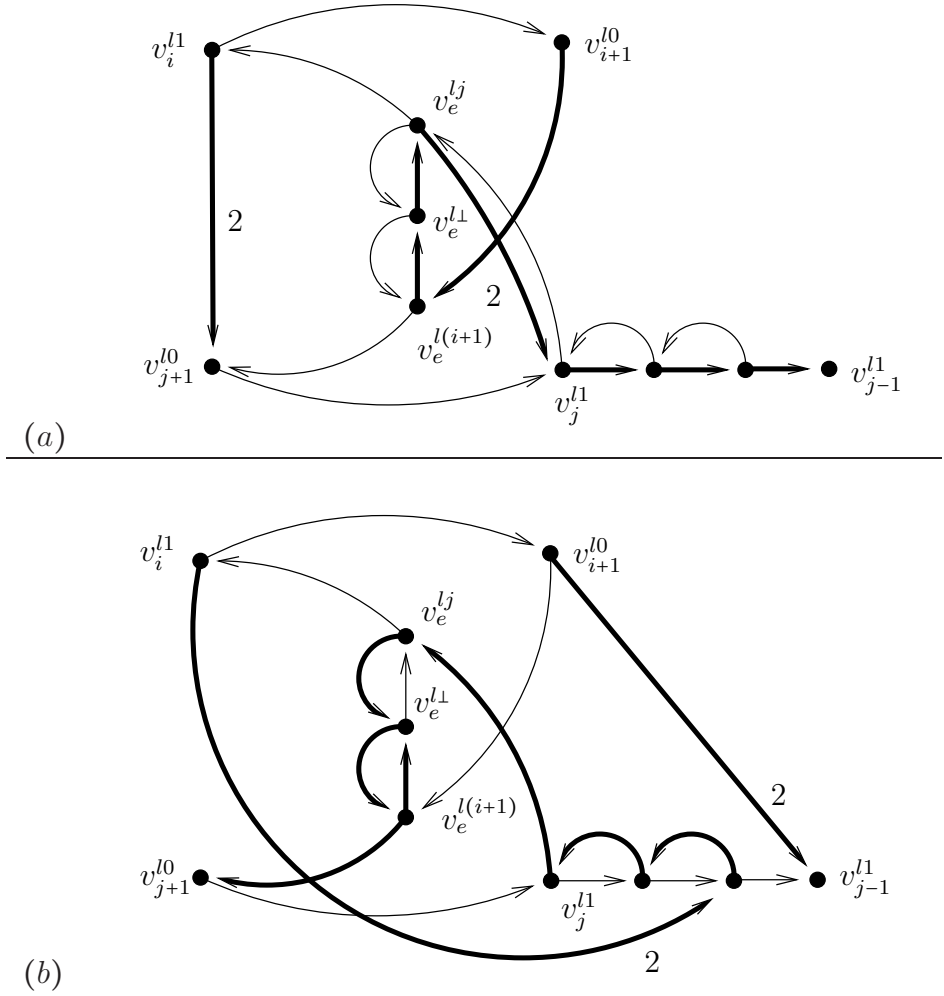

\begin{center}
\begin{tabular}[c]{lc}
$(a)$ & \input{figures/fig12atsp28.pspdftex} \\
\hline
$(b)$ &  \input{figures/fig12atsp29.pspdftex}
\end{tabular}
\end{center}
\caption{6.Case  with  $\psi_{\sigma}(x_i)\neq \psi_{\sigma}(x_{i+1})=0$ and 
$\psi_{\sigma}(x_j) \neq \psi_{\sigma}(x_{j+1})=1$.  }
\label{fig:12atsp:19}
\end{figure} 
\noindent
The case, in which 
$\psi_{\sigma}(x_i) \oplus \psi_{\sigma}(x_{i+1})=1$, 
$\psi_{\sigma}(x_i) \oplus \psi_{\sigma}(x_{j})=1$ and 
$\psi_{\sigma}(x_j) \oplus \psi_{\sigma}(x_{j+1})=0$ holds, can be 
discussed analogously.\\
\\
%
We obtain the following proposition.
\begin{proposition}\label{12atspsigtopsimatching}
Let   $x^l_i \oplus x^l_{i+1}=0$, $x^l_i \oplus x^l_{j}=0$ and 
$x^l_j \oplus x^l_{j+1}=0$ be  equations in $\cH$.
Then, it is possible to transform in polynomial time 
the given tour $\sigma$ 
passing through the graphs corresponding to $x^l_i \oplus x^l_{i+1}=0$, $x^l_i \oplus x^l_{j}=0$ 
and $x^l_j \oplus x^l_{j+1}=0$ such that it  has local length 
$5 +u$ and 
 the number of unsatisfied equations in $\{x^l_i \oplus x^l_{i+1}=0, 
x^l_i \oplus x^l_{j}=0,x^l_j \oplus x^l_{j+1}=0\}$ by $\psi_{\sigma}$ is bounded
from above by $u$.
\end{proposition}   
\subsubsection{Transforming $\sigma$ in Graphs Corresponding to Equations With Three Variables}
Let $g^3_c\equiv x^l_i \oplus x^s_j \oplus x^r_k = 0$ be an equation with three variables  in $\cH$.
Furthermore, let $\cC_l$ be a circle in $\cH$ and $x^l_i \oplus x^l_{i+1}=0$ a circle equation.
For notational simplicity, we set $e=\{i,i+1\}$. We are going to analyze the the number of satisfied equations
by $\psi_{\sigma}$ in relation to 
the local length of $\sigma$ in the graphs $P^l_{i}$, $P^l_{i+1}$, $P^l_{e}$
and $D^3_{c}$. First, we transform the tour traversing the graphs $P^l_{i}$, $P^l_{i+1}$ and $P^l_{e}$ such that
it uses the $\psi_{\sigma}(x^l_i)$-traversal of $P^l_{e}$.
Afterwards, due to the construction of $D^3_c$ and Proposition~\ref{pro:gadget3atsp}, 
the tour can be transformed such that it has  local length
$3\cdot3+4$ if it passes an even number of parity graphs  $P\in \{P^l_{e},P^k_{a},P^s_{b}\}$ 
by using a simple path
through $D^3_c$, otherwise, it yields a local length of $13+1$. 
Let us start with the case $\psi_{\sigma }(x^l_i)=1$ and $ \psi_{\sigma }(x^l_{i+1})=1$.\\
\\

\begin{figure}[h!]
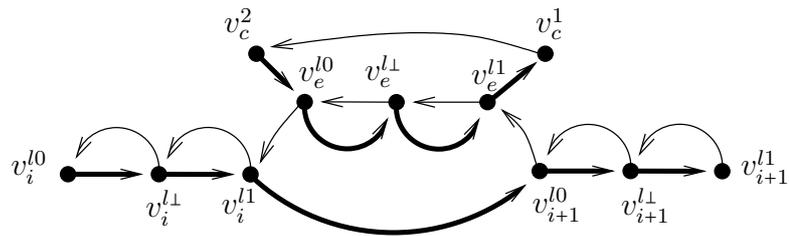

\begin{center}
\begin{tabular}[t]{lc}
$(a)$ &  \input{figures/fig12atsp40.pspdftex}  \\
 \hline  \\
$(b)$ & \input{figures/fig12atsp40b.pspdftex}
\end{tabular}
\end{center}
\caption{1. Case $\psi_{\sigma }(x^l_i)=1$ and $ \psi_{\sigma }(x^l_{i+1})=1$ }
\label{fig:12atsp:25}
\end{figure} 
%
\noindent
\textbf{1. Case $\psi_{\sigma }(x^l_i)=1$ and $ \psi_{\sigma }(x^l_{i+1})=1$:}\\
In Figure~\ref{fig:12atsp:25} $(a)$ and $(b)$, we display the tour passing through $P^l_{i}$,  $P^l_{i+1}$
 and $P^l_{e}$ with  $\psi_{\sigma }(x^l_i)=1$
and $ \psi_{\sigma }(x^l_{i+1})=1$ before and after the transformation, respectively.\\
\\
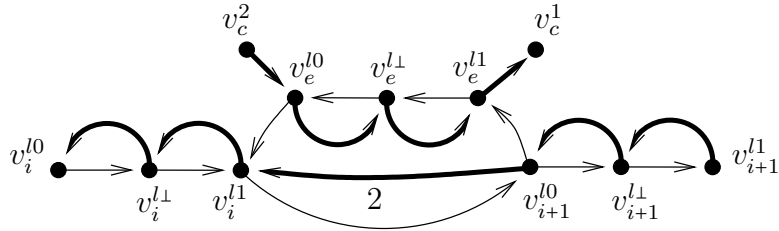
\begin{figure}[h!]
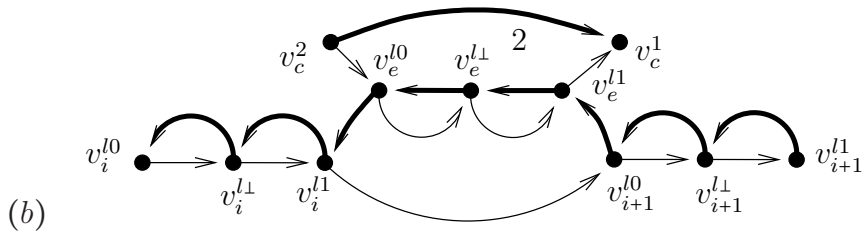

\begin{center}
\begin{tabular}[t]{lc}
$(a)$ &  \input{figures/fig12atsp41b.pspdftex}  \\
 \hline  \\
$(b)$ & \input{figures/fig12atsp41.pspdftex}
\end{tabular}
\end{center}
\caption{2. Case $\psi_{\sigma }(x^l_i)=0$ and  $ \psi_{\sigma }(x^l_{i+1})=0$ }
\label{fig:12atsp:26}
\end{figure}

It is possible to transform the tour $\sigma$ without increasing the length such that it traverses 
the arc $(v^{l1}_i,v^{l0}_{i+1})$.
In the outer loop, the tour may use at least one of the arcs $(v^2_c, v^{l0}_e)$ and $( v^{l1}_e,v^1_c)$
depending on the parity check in $D^3_c$. We associate the local length $1$ with this part of the tour.\\
\\

\begin{figure}[h!]
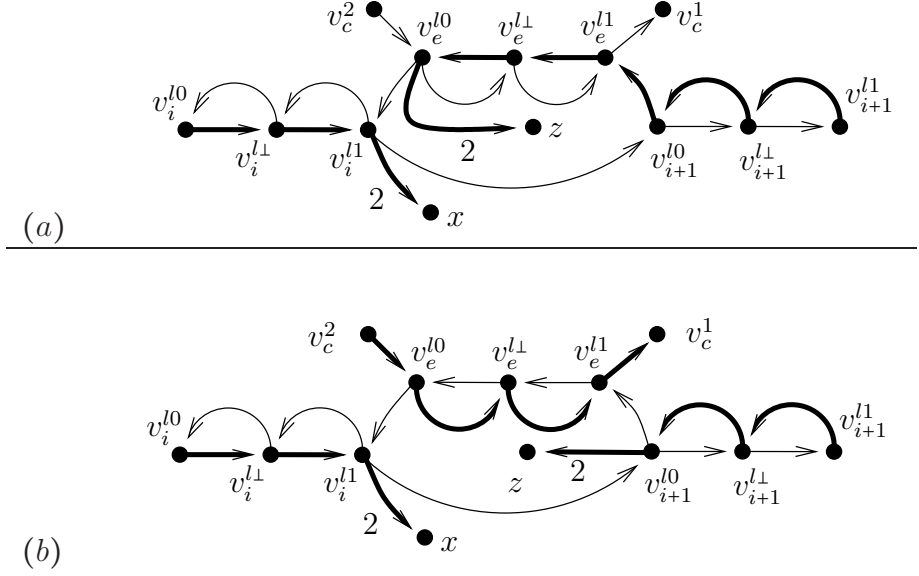

\begin{center}
\begin{tabular}[t]{lc}
$(a)$ &  \input{figures/fig12atsp42a.pspdftex}  \\
 \hline  \\
$(b)$ & \input{figures/fig12atsp42b.pspdftex}
\end{tabular}
\end{center}
\caption{3. Case $\psi_{\sigma }(x^l_i)=1$, $ \psi_{\sigma }(x^l_{i+1})=0$ 
and $\psi_{\sigma }(x^l_i)\oplus \psi_{\sigma }(x^s_j) \oplus \psi_{\sigma }(x^r_k) =0$  }
\label{fig:12atsp:27}
\end{figure}
\begin{figure}[h!]
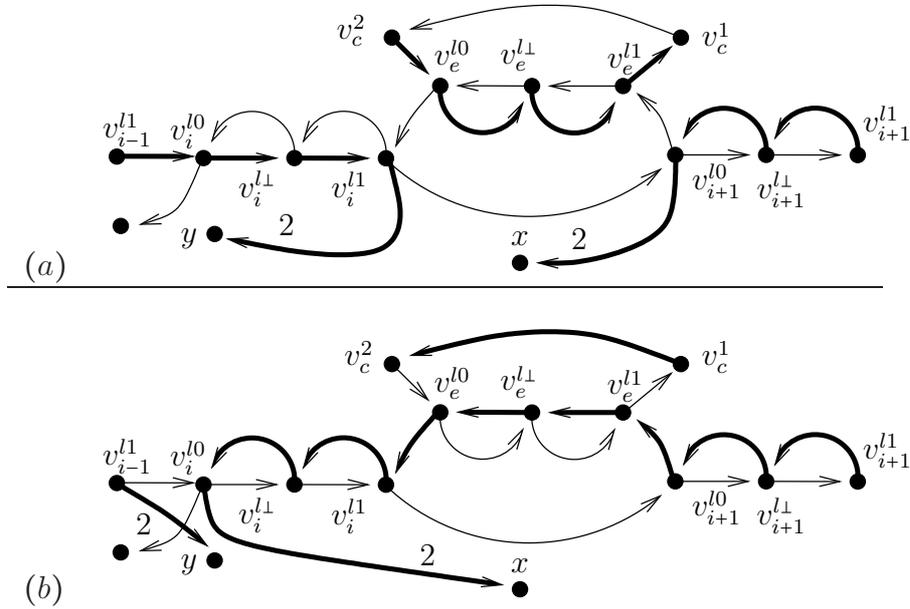

\begin{center}
\begin{tabular}[t]{lc}
$(a)$ &  \input{figures/fig12atsp43.pspdftex}  \\
 \hline  \\
$(b)$ & \input{figures/fig12atsp43b.pspdftex}
\end{tabular}
\end{center}
\caption{3. Case $\psi_{\sigma }(x^l_i)=1$, $ \psi_{\sigma }(x^l_{i+1})=0$ 
and $\psi_{\sigma }(x^l_i)\oplus \psi_{\sigma }(x^s_j) \oplus \psi_{\sigma }(x^r_k) =1$.  }
\label{fig:12atsp:28}
\end{figure}  
\noindent
\textbf{2. Case $\psi_{\sigma }(x^l_i)=0$ and  $ \psi_{\sigma }(x^l_{i+1})=0$ :}\\
%
%
%
In Figure~\ref{fig:12atsp:26}, we display the underlying scenario with   $\psi_{\sigma }(x^l_i)=0$
and $ \psi_{\sigma }(x^l_{i+1})=0$.
The transformed tour uses the $0$-traversal of the parity graph $P^l_e$. The vertices
$v^2_c$ and $v^1_c$ are connected via a $2$-arc. We assign the local length $1$ to this part
of the tour.\\
\\
%
%
 %
%
 %
%
\noindent
\textbf{3. Case $\psi_{\sigma }(x^l_i)=1$ and  $ \psi_{\sigma }(x^l_{i+1})=0$ :}\\
Let us assume that $\psi_{\sigma }(x^l_i)\oplus \psi_{\sigma }(x^s_j) \oplus \psi_{\sigma }(x^r_k) =0$ holds.
Hence, it is possible to transform the tour such that it uses the path 
$v^2_c\rightarrow v^{l0}_e \rightarrow v^{l\bot}_e \rightarrow v^{l1}_e \rightarrow v^1_c$
and thus, the $0$-traversal of the parity graph $P^l_e$ as displayed in  Figure~\ref{fig:12atsp:27}.

\begin{figure}[h!]
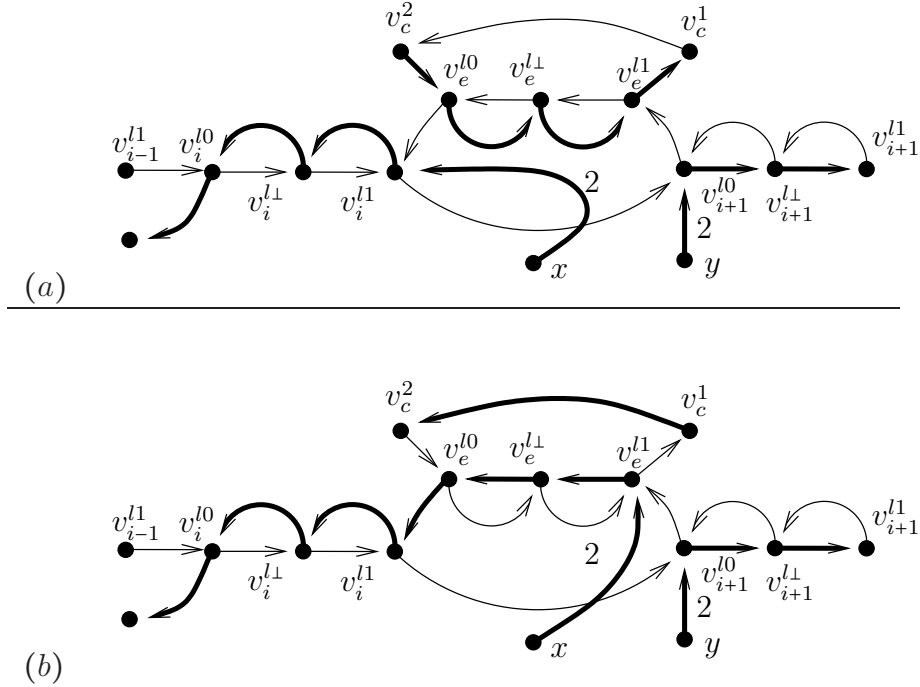

\begin{center}
\begin{tabular}[t]{lc}
$(a)$ &  \input{figures/fig12atsp44.pspdftex}  \\
 \hline  \\
$(b)$ & \input{figures/fig12atsp44b.pspdftex}
\end{tabular}
\end{center}
\caption{4. Case $\psi_{\sigma }(x^l_i)=0$, $ \psi_{\sigma }(x^l_{i+1})=1$
and 
$\psi_{\sigma }(x^l_i)\oplus \psi_{\sigma }(x^s_j) \oplus \psi_{\sigma }(x^r_k) =0$ }
\label{fig:12atsp:29}
\end{figure} 

\noindent
In the other case, namely  $\psi_{\sigma }(x^l_i)\oplus \psi_{\sigma }(x^s_j) 
\oplus \psi_{\sigma }(x^r_k) =1$,
we will change the value of $\psi_{\sigma }(x^l_i)$ achieving in this way at least 
$2-1$ more satisfied equation. Let us examine the scenario in Figure~\ref{fig:12atsp:28}. 
The tour uses the $0$-traversal of the parity graph $P^l_e$, which enables $\sigma$
to pass the parity check in $D^3_c$.

In both cases, we obtain the local length $2$ in conformity with the at most one unsatisfied equation
by $\psi_{\sigma}$.
 
\begin{figure}[h]
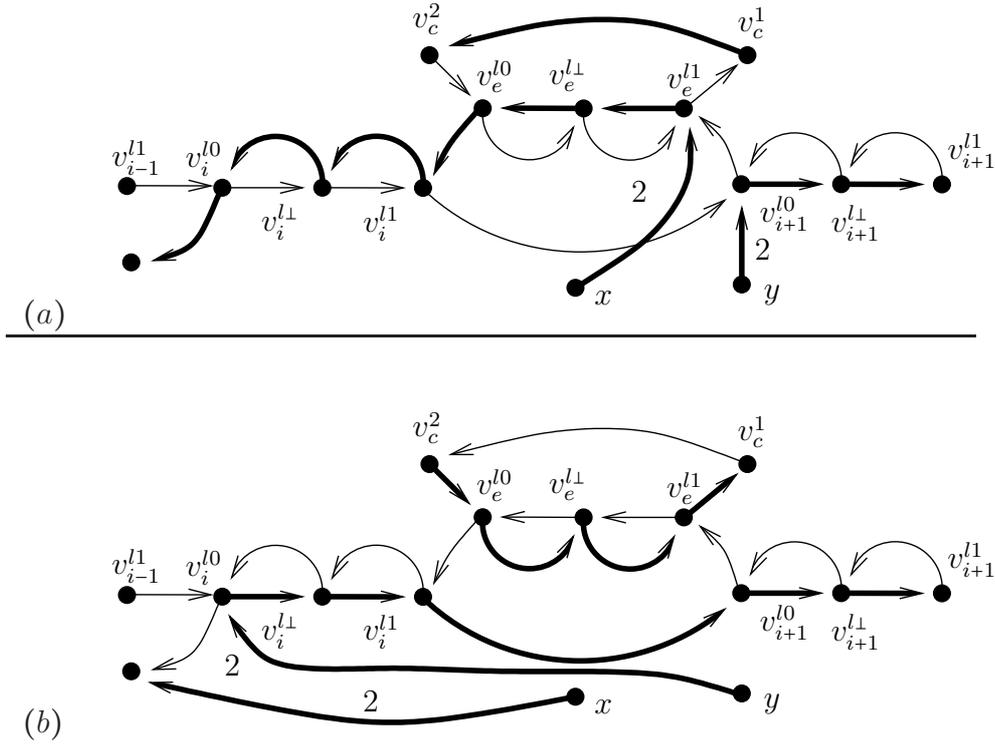

\begin{center}
\begin{tabular}[t]{lc}
$(a)$ &  \input{figures/fig12atsp45.pspdftex}  \\
 \hline  \\
$(b)$ & \input{figures/fig12atsp45b.pspdftex}
\end{tabular}
\end{center}
\caption{4. Case $\psi_{\sigma }(x^l_i)=0$,  $ \psi_{\sigma }(x^l_{i+1})=1$ and 
$\psi_{\sigma }(x^l_i)\oplus \psi_{\sigma }(x^s_j) \oplus \psi_{\sigma }(x^r_k) =1$ }
\label{fig:12atsp:30}
\end{figure} 
 %

\noindent
\textbf{4. Case $\psi_{\sigma }(x^l_i)=0$ and  $ \psi_{\sigma }(x^l_{i+1})=1$ :}\\
Assuming $\psi_{\sigma }(x^l_i)\oplus \psi_{\sigma }(x^s_j) \oplus \psi_{\sigma }(x^r_k) =0$
and the scenario depicted in Figure~\ref{fig:12atsp:29} $(a)$,
the tour will be modified such that the parity graphs $P^l_i$ and $P^l_e$ are traversed in 
the same direction. Since we have $\psi_{\sigma }(x^l_i)\oplus \psi_{\sigma }(x^s_j) \oplus 
\psi_{\sigma }(x^r_k) =0$,
we are able to uncouple the parity graph $P^l_e$ from the tour through $D^3_c$ without increasing its length.
We display the transformed tour in Figure~\ref{fig:12atsp:29} $(b)$.

Assuming $\psi_{\sigma }(x^l_i)\oplus \psi_{\sigma }(x^s_j) \oplus \psi_{\sigma }(x^r_k) =1$
and the scenario depicted in Figure~\ref{fig:12atsp:30} $(a)$,
we transform $\sigma$ such that the parity graph $P^l_e$ is traversed when $\sigma$ is passing through
$D^3_c$ meaning $v^2_c\rightarrow v^{l0}_e \rightarrow v^{l\bot}_e \rightarrow v^{l1}_e \rightarrow v^1_c$
is a part of the tour. In addition, we change the value of $\psi_{\sigma}(x^l_i)$ yielding at least $2-1$
more satisfied equations. The transformed tour is displayed in Figure~\ref{fig:12atsp:30} $(b)$.

In both cases, we associate the local length $2$ with $\sigma$. On the other hand, 
$\psi_{\sigma}$ leaves at most one equation unsatisfied.\\
\\
We obtain the following proposition. 
\begin{proposition}\label{12atspsigtopsieq3}
Let $g^3_c\equiv x^l_i \oplus x^s_j \oplus x^r_k = 0$ be an equation with three variables  in $\cH$.
Furthermore, let  $x^l_i \oplus x^l_{i+1}=0$, 
$x^s_j \oplus x^s_{j+1}=0$ and $x^r_k \oplus x^r_{k+1}=0$ be circle equations in $\cH$.
Then, it is possible to transform in polynomial time 
the given tour $\sigma$
passing through the graph corresponding to $x^l_i \oplus x^l_{i+1}=0$, 
$x^s_j \oplus x^s_{j+1}=0$, $x^r_k \oplus x^r_{k+1}=0$ and $g^3_c$ such that it has local length 
$4+3\cdot3 +3 +u$ and
 the number of unsatisfied equations in $\{x^l_i \oplus x^l_{i+1}=0,x^s_j \oplus x^s_{j+1}=0,
x^r_k \oplus x^r_{k+1}=0, g^3_c\}$ by $\psi_{\sigma}$ is at most $u$.
\end{proposition}

\subsubsection{Transforming $\sigma$ in Graphs Corresponding to Circle Border Equations}
Let $C_l$ be a circle in $\cH$ and $x^l_1\oplus x^l_n=0$ its
circle border equation. Furthermore, let $g^3_c\equiv x^l_n \oplus x^s_j \oplus x^r_k = 0 $ be an 
equation with three
variables  contained in $\cH$. We are going to transform a given tour $\sigma$
passing through the graph corresponding to $x^l_1\oplus x^l_n=0$ such that it will have the local 
length $2$ if $x^l_1\oplus x^l_n=0$ is satisfied by $\psi_\sigma$ and $3$, otherwise.  
  In each case, we modify $\sigma$ such that it
uses a $\psi(x^l_n)$-traversal of $P^l_{\{1,n\}}$.  Afterwards, $\sigma$ will be checked  in $D^3_c$ 
whether it passes the parity test. Let us begin with the analysis starting with the case 
$\psi_{\sigma }(x_1)=0$ and $ \psi_{\sigma }(x_n)=0$.

\begin{figure}[h!]
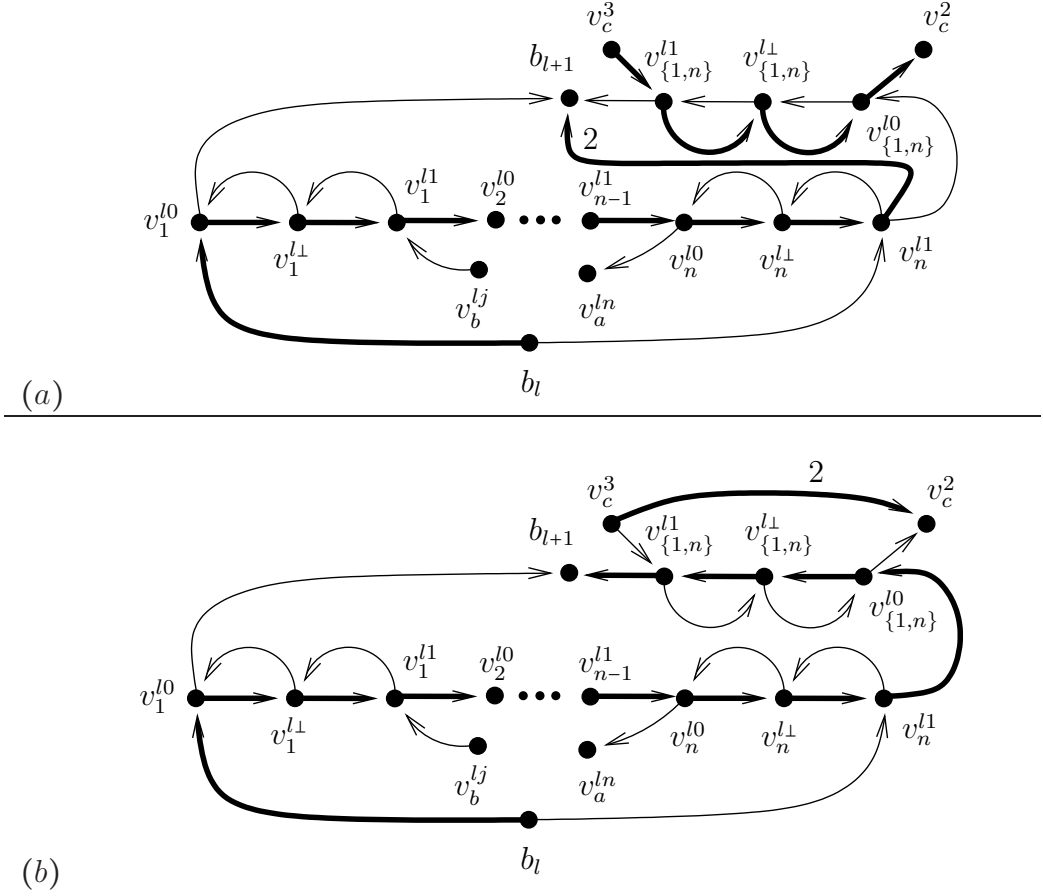

\begin{center}
\begin{tabular}[t]{lc}
$(a)$ &  \input{figures/fig12atsp35.pspdftex}  \\
 \hline  \\
$(b)$ & \input{figures/fig12atsp35b.pspdftex}
\end{tabular}
\end{center}
\caption{Case $\psi_{\sigma }(x_1)=0$ and $ \psi_{\sigma }(x_n)=0$ }
\label{fig:12atsp:21a}
\end{figure}
 
\noindent
\textbf{1. Case $\psi_{\sigma }(x_1)=0$ and $ \psi_{\sigma }(x_n)=0$:}\\
Let us assume that $\psi_{\sigma}$ leaves $g^3_c$ 
unsatisfied meaning $\psi_{\sigma }(x^l_n)\oplus \psi_{\sigma }(x^s_j) \oplus \psi_{\sigma }(x^r_k) =1$. 
In addition to it, we assume that the path $v^3_c \rightarrow v^{l1}_{\{1,n\}}
\rightarrow v^{l\bot}_{\{1,n\}} \rightarrow  v^{l0}_{\{1,n\}}
\rightarrow v^2_c$ is a part of $\sigma$. Notice that $\sigma$ fails the parity check in  $D^3_c$
if $v^3_c \rightarrow v^{l1}_{\{1,n\}}
\rightarrow v^{l\bot}_{\{1,n\}} \rightarrow  v^{l0}_{\{1,n\}}
\rightarrow v^2_c$ is not used by $\sigma$. 
First, 
 we modify the tour such that it includes the arc $(b_l, v^{l0}_1)$. 
 For the same reason, we may assume that $v^{l1}_n$
and $b_{l+1}$ is connected via a $2$-arc.
We obtain the scenario depicted in Figure~\ref{fig:12atsp:21a} $(a)$.
As for the next step, we transform $\sigma$ such that it contains the arcs 
$( v^{l1}_n,v^{l0}_{\{1,n\}} )$ and $(v^{l1}_{\{1,n\}}, b_{l+1})$. Consequently, we use the $1$-traversal of  the 
parity graph $P^l_{\{1,n\}}$ and connect $v^3_c$ and $v^2_c$ via a $2$-arc.
The modified tour is depicted in Figure~\ref{fig:12atsp:21a} $(b)$. 

If $\psi_{\sigma}$ satisfies $g^3_c$ and $\sigma$ contains the path 
$v^3_c \rightarrow v^{l1}_{\{1,n\}}
\rightarrow v^{l\bot}_{\{1,n\}} \rightarrow  v^{l0}_{\{1,n\}}
\rightarrow v^2_c$, 	we modify $\sigma$ in $D^3_c$ such that it passes
the parity test in $D^3_c$ and contains the arc $(v^2_c, v^3_c)$. 

In both cases, we associate the local length $2$
with this part of $\sigma$.\\
\\

\begin{figure}[h!]
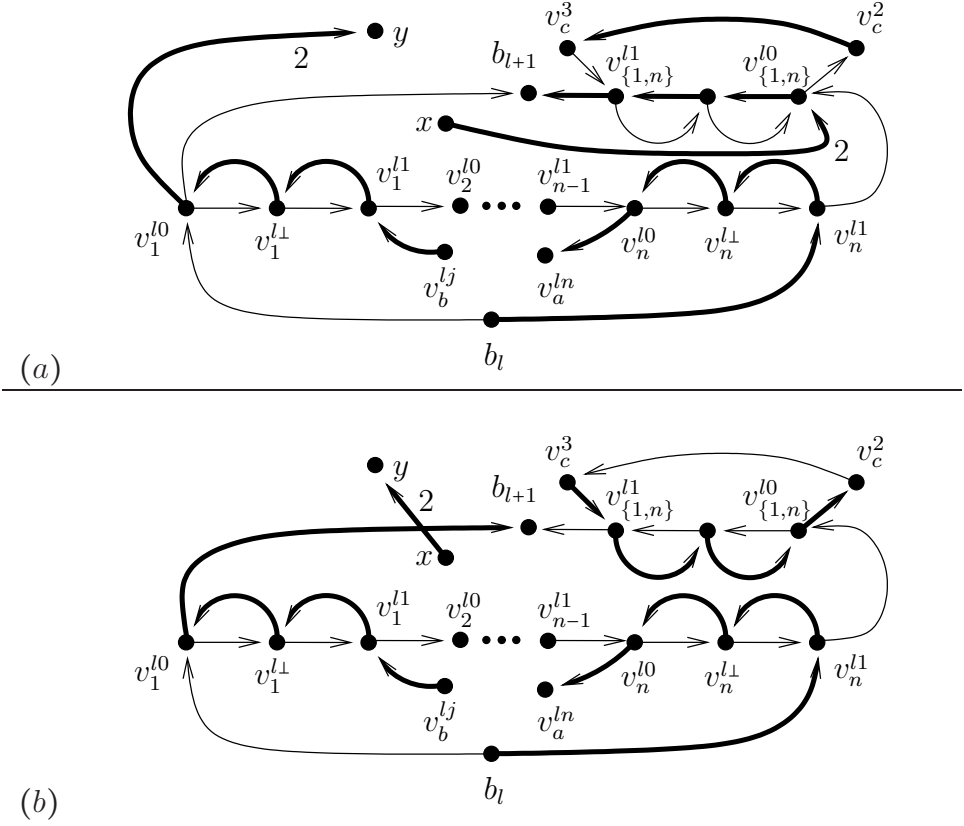

\begin{center}
\begin{tabular}[t]{lc}
$(a)$ &  \input{figures/fig12atsp36.pspdftex}  \\
 \hline  \\
$(b)$ & \input{figures/fig12atsp36b.pspdftex}
\end{tabular}
\end{center}
\caption{2. Case $\psi_{\sigma }(x_1)=1$ and $ \psi_{\sigma }(x_n)=1$ }
\label{fig:12atsp:22}
\end{figure} 

\noindent
\textbf{2. Case $\psi_{\sigma }(x_1)=1$ and $ \psi_{\sigma }(x_n)=1$:}\\
Let us assume that $\psi_{\sigma }(x^l_n)\oplus \psi_{\sigma }(x^s_j) \oplus \psi_{\sigma }(x^r_k) =1$
holds and $\sigma$ contains the arc $(v^2_c, v^3_c)$.  
Given this scenario, we may assume that $(b_l, v^{l1}_n)$ is contained in $\sigma$ due to a simple
modification. Then, we are going to analyze the situation depicted in Figure~\ref{fig:12atsp:22}~$(a)$.
We transform $\sigma$ in the way described in Figure~\ref{fig:12atsp:22}~$(b)$. Afterwards, $\sigma$ will
be  modified
in $D^3_c$ such that it uses a simple  path in $D^3_c$ failing the parity check.

The case, in which  $\psi_{\sigma }(x^l_i)\oplus \psi_{\sigma }(x^l_i) \oplus \psi_{\sigma }(x^l_i) =0$
holds
and $\sigma$ contains the arc $(v^2_c, v^3_c)$, can be discussed similarly since $\sigma$
passes the parity check by including the path $v^3_c \rightarrow v^{l1}_{\{1,n\}}
\rightarrow v^{l\bot}_{\{1,n\}} \rightarrow  v^{l0}_{\{1,n\}}
\rightarrow v^2_c$. 

In both cases, we associate the length $2$
with this part of $\sigma$.\\
\\
\begin{figure}[h!]
\begin{center}
\begin{tabular}[t]{lc}
$(a)$ &  \input{figures/fig12atsp37.pspdftex}  \\
 \hline  \\
$(b)$ & \input{figures/fig12atsp37b.pspdftex} \\
\hline\\
$(c)$ & \input{figures/fig12atsp37c.pspdftex} 
\end{tabular}
\end{center}
\caption{3. Case $\psi_{\sigma }(x_1)=0$ and $ \psi_{\sigma }(x_n)=1$ }
\label{fig:12atsp:23}
\end{figure}   

\noindent
\textbf{3. Case $\psi_{\sigma }(x_1)=0$ and $ \psi_{\sigma }(x_n)=1$:}\\
Let us assume that $\psi_{\sigma }(x^l_n)\oplus \psi_{\sigma }(x^s_j) \oplus \psi_{\sigma }(x^r_k) =1$
holds and $\sigma$ traverses the path $v^3_c \rightarrow v^{l1}_{\{1,n\}}
\rightarrow v^{l\bot}_{\{1,n\}} \rightarrow  v^{l0}_{\{1,n\}}
\rightarrow v^2_c$. 
Then, we transform the tour $\sigma$ such that it contains the arc 
  $(v^{l0}_{1},b_{l+1})$. Note that neither $(b_l, v^{l0}_1)$ nor $(b_l, v^{l1}_n)$
  is included in the tour. Hence, $\sigma$ contains  a $2$-arc to connect $b_l$.
  The same holds for the vertex $v^{l1}_n$. 
  This situation is displayed in Figure~\ref{fig:12atsp:23} $(a)$.

We are going to invert the value of $\psi_{\sigma}(x^l_n)$ such that $\psi_{\sigma}$
  satisfies $g^3_c$ and $x^l_1 \oplus x^l_n =0$. In this way, we gain at least $2-1$
  more satisfied equations. The corresponding transformation is pictured in  
   Figure~\ref{fig:12atsp:23} $(b)$.
   
On the other hand, if we assume that 
$\psi_{\sigma }(x^l_n)\oplus \psi_{\sigma }(x^s_j) \oplus \psi_{\sigma }(x^r_k) =0$
holds and $\sigma$ traverses the path $v^3_c \rightarrow v^{l1}_{\{1,n\}}
\rightarrow v^{l\bot}_{\{1,n\}} \rightarrow  v^{l0}_{\{1,n\}}
\rightarrow v^2_c$, we modify the tour as depicted in Figure~\ref{fig:12atsp:23} $(c)$.
Note that $\sigma$ passes the parity check in $D^3_c$ and therefore, the tour
may use a simple path in $D^3_c$.

We associate the local length $3$ with $\sigma$ in this case.\\
\\

\begin{figure}[h!]
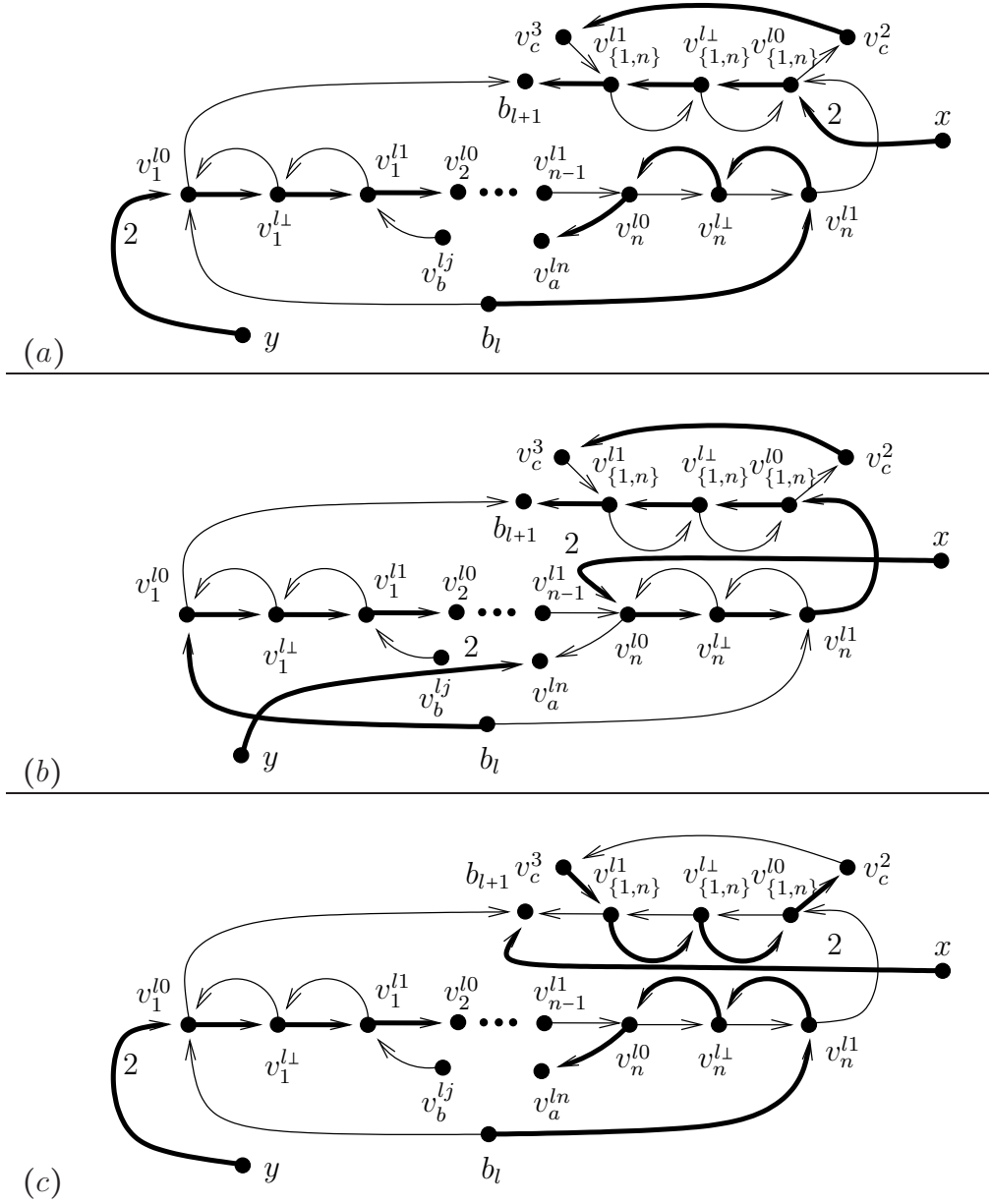

\begin{center}
\begin{tabular}[t]{lc}
$(a)$ &  \input{figures/fig12atsp38.pspdftex}  \\
 \hline  \\
$(b)$ & \input{figures/fig12atsp38b.pspdftex} \\
\hline \\
$(c)$ & \input{figures/fig12atsp38c.pspdftex}
\end{tabular}
\end{center}
\caption{4. Case $\psi_{\sigma }(x_1)=1$ and $ \psi_{\sigma }(x_n)=0$ }
\label{fig:12atsp:24}
\end{figure}  
\noindent
\textbf{4. Case $\psi_{\sigma }(x_1)=1$ and $ \psi_{\sigma }(x_n)=0$:}\\
Let us assume that $\psi_{\sigma }(x^l_n)\oplus \psi_{\sigma }(x^s_j) \oplus \psi_{\sigma }(x^r_k) =1$
holds and $\sigma$ uses the arc  $(v^2_c ,v^3_c)$. 

Then, we transform the tour $\sigma$ such that it contains the arc 
  $(b_{l}, v^{l1}_{n})$. We note that neither $( v^{l0}_1, b_{l+1})$ nor 
  $( v^{l1}_n,v^{l0}_{\{1,n\}} )$
  is included in the tour. For this reason, $\sigma$ must use  a $2$-arc to connect $v^{l0}_{\{1,n\}}$.
  The same holds for the vertex $v^{l0}_1$.
  The corresponding situation is displayed in Figure~\ref{fig:12atsp:24} $(a)$.
  
We modify the tour as displayed in  Figure~\ref{fig:12atsp:24} $(b)$ and obtain at least $2-1$
more satisfied equations. 
  
On the other hand, if we assume that  
$\psi_{\sigma }(x^l_n)\oplus \psi_{\sigma }(x^s_j) \oplus \psi_{\sigma }(x^r_k) =0$
holds, we need to include the path $v^3_c \rightarrow v^{l1}_{\{1,n\}}
\rightarrow v^{l\bot}_{\{1,n\}} \rightarrow  v^{l0}_{\{1,n\}}
\rightarrow v^2_c$ as depicted in Figure~\ref{fig:12atsp:24} $(c)$.\\
In both cases, we associate the local length $3$ with this part of the tour.\\
\\
We obtain the following statement.
\begin{proposition}\label{12atspsigtopsicirclebord}
Let $C_l$ be a circle in $\cH$ and $x^l_1\oplus x^l_n=0$ its
circle border equation. Then, it is possible to transform in polynomial time 
the given tour $\sigma$
passing through the graph corresponding to $x^l_1\oplus x^l_n=0$ such that it has local length $2$
if $x^l_1\oplus x^l_n=0$ is satisfied  by $\psi_{\sigma}$ and $3$, otherwise.
\end{proposition}
\noindent
Thus far, we are ready to give the proof of Theorem~\ref{thm:main12atsp}$(i)$.

\subsection{Proof of Theorem~\ref{thm:main12atsp}$(i)$}
Let $\cH$ be an  instance of the Hybrid problem consisting of $n$ circles $\cC_1, \ldots, \cC_n$,
$m_2$ equations with two variables
and $m_3$ equations with three variables $g^3_c$ with $c\in [m_3]$. Then, we construct 
in polynomial time the corresponding instance $D_{\cH}=(V(D_{\cH}),A(D_{\cH}))$ of the
$(1,2)$-ATSP problem as described in Section~\ref{sec:constrinst(12)atsp}.\\
\\
$(a)$~Let $\phi$ be an assignment  to the variables in $\cH$ leaving $u$ equations 
in $\cH$ unsatisfied.
According to Proposition~\ref{prep:phitosigmat}~--~\ref{prep:phitosigcirlb}, 
it is possible to construct in polynomial time the  tour $\sigma_{\phi}$ with length 
$$\ell(\sigma_{\phi})\leq  3\cdot m_2+ (4+3\cdot 3)\cdot m_3+n+1+u.   $$ 
\noindent
$(b)$ ~Let $\sigma$ be  a tour  in  $D_{\cH}$ with 
length $\ell(\sigma)=3\cdot m_2+ 13\cdot m_3+n+1+u$.
Due to  Proposition~\ref{prop:consistent} we may assume that $\sigma$
 uses only $0/1$-traversals of every parity graph included in $D_{\cH}$.
According to Definition~\ref{def:psi}, we associate the corresponding assignment $\psi_{\sigma}$
with the underlying tour $\sigma$.   
Recall from  Proposition~\ref{12atspsigtopsimatching}~--~\ref{12atspsigtopsicirclebord} that 
 it is possible to convert
$\sigma$ in polynomial time into a tour $\sigma'$ 
 without increasing the length such that
$\psi_{\sigma'}$ leaves at most $u$ equations in $\cH$ unsatisfied. \qed
\section{Approximation Hardness of the $(1,4)$-ATSP Problem   }
In order to prove the claimed hardness results for the $(1,4)$-ATSP problem, we use the same 
construction described in Section~\ref{sec:constrinst(12)atsp} 
with the difference that all arcs in parity graphs have weight $1$,
whereas all other arcs contained in the directed graph $D_{\cH}$ obtain the weight $2$. 
The induced asymmetric metric space $(V_{\cH},d_{\cH})$ is given by $V_{\cH}=V(D_{\cH})$ and  distance
function defined by the shortest path metric in $D_{\cH}$ bounded by the value $4$.
In other words, given $x,y\in V_{\cH}$, the distance between $x$ and $y$ in $V_{\cH}$ is
$$d_{\cH}(x,y)=\min\{\textrm{length of a shortest path from $x$ to $y$ in $D_{\cH}$}, ~4 \}.$$ 
The only difficulty that remains is to prove that tours remain consistent.
Thus, we have to prove that given a tour $\sigma$ in $V_{\cH}$, we are able to transform
$\sigma$ in polynomial time into a tour $\sigma'$, which uses only $0/1$-traversals in parity graphs 
contained in $D_{\cH}$, without increasing $\ell(\sigma)$. This statement 
can be proved by considering all possibilities exhaustively. 
 Some cases are displayed in Figure~\ref{figcases14atsp1} -- Figure~\ref{figcases14atsp3}.\\
\\
We are ready to give the proof of Theorem~\ref{thm:main12atsp}~$(ii)$.

\subsection{Proof of Theorem~\ref{thm:main12atsp}~$(ii)$ }
Given $\cH$ an  instance of the Hybrid problem consisting of $n$ circles $\cC_1, \ldots, \cC_n$,
$m_2$ equations with two variables
and $m_3$ equations with three variables $g^3_c$ with $c\in [m_3]$,
 we construct in polynomial time 
the associated instance $(V_{\cH},d)$ of the
$(1,4)$-ATSP problem.\\ 
Given an assignment $\phi$ to the variables of $\cH$ 
leaving $u$ equations unsatisfied in $\cH$, then there exists a tour 
with  length at most $ m_2\cdot (2+ 2) + m_3\cdot ( 3\cdot 4+2\cdot 4 )+2\cdot u +2(n+1)  $.\\ 
On the other hand, if we are given a tour $\sigma$ in $V_{\cH}$ with length $4m_2+20m_3+2n +2+2\cdot u $, it is 
possible to transform $\sigma$ in polynomial time into a tour $\sigma'$ without increasing the 
length such that  
the associated assignment $\psi_{\sigma'}$  leaves
at most $u$ equations in $\cH$ unsatisfied. \qed   
\section{Approximation Hardness of the $(1,2)$-TSP Problem }
In order to prove Theorem~\ref{thm:main12tsp}~$(i)$, we apply the 
reduction method used in Section~\ref{sec:12atsp} to the $(1,2)$-TSP problem.
As for the parity gadget, we use the graph  depicted in Figure~\ref{fig:12tsp:1} with its corresponding
traversals. The  traversed edges are pictured by thick lines.\\
 \\
\begin{figure}[h!]
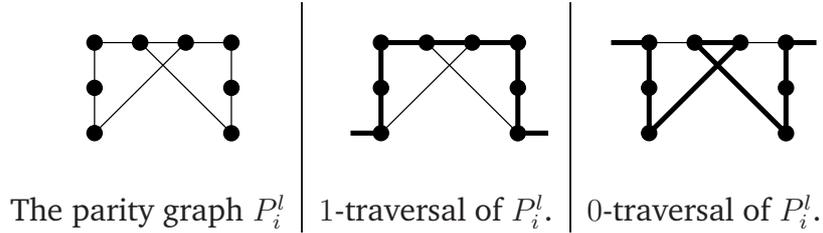

\begin{center}
\begin{tabular}[c]{c|c|c}
\input{figures/fig12parity.pspdftex} & \input{figures/fig12tsp1a.pspdftex} & \input{figures/fig12tsp1b.pspdftex} \\
 The parity graph $P^l_i$ & $1$-traversal of $P^l_i$.   & $0$-traversal of $P^l_i$.
 \end{tabular}
\end{center}
\caption{Traversal of the graph $P^l_i$ given the assignment $\phi$. }
\label{fig:12tsp:1}
\end{figure} 
\begin{figure}[h!]
\begin{center}
\input{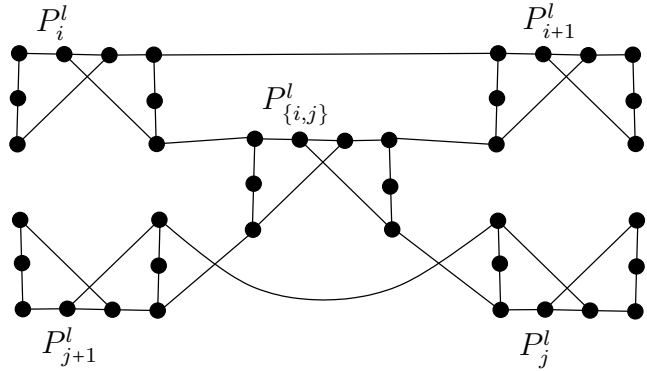}
\end{center}
\caption{Graphs corresponding to  equations $x^l_i \oplus x^l_j =0$,  
$x^l_i \oplus x^l_{i+1} =0$ \& $x^l_j \oplus x^l_{j+1} =0$. }
\label{fig:12tsp:3}
\end{figure} 

Let $\cH$ be an instance of the hybrid problem.
Given a matching equation $x^l_i \oplus x^l_j =0$ in $\cH$ and the corresponding  
circle equations
$x^l_i \oplus x^l_{i+1} =0$ and $x^l_j \oplus x^l_{j+1} =0$, we connect the associated 
parity graphs $P^l_{i}$, $P^l_{i+1}$, $P^l_{\{i,j\}}$, $P^l_{j}$ and $P^l_{j+1}$
 as displayed in Figure~\ref{fig:12tsp:3}.\\
\\
For equations with three variables $g^3_c\equiv x\oplus y \oplus z=0$ 
in $\cH$,  we use the graph $G^3_c$ depicted in 
Figure~\ref{fig:12tsp:2}. Recall from Proposition~\ref{pro:gadget3tsp} that  
there is a simple path from    $s_c$ to  $s_{c+1}$ in Figure~\ref{fig:12tsp:2} 
 containing the vertices 
$v\in \{v^1_c, v^2_c\}$  if and
only if an even number of parity graphs is  traversed.\\
\\ 

\begin{figure}[h!]
\begin{center}
\input{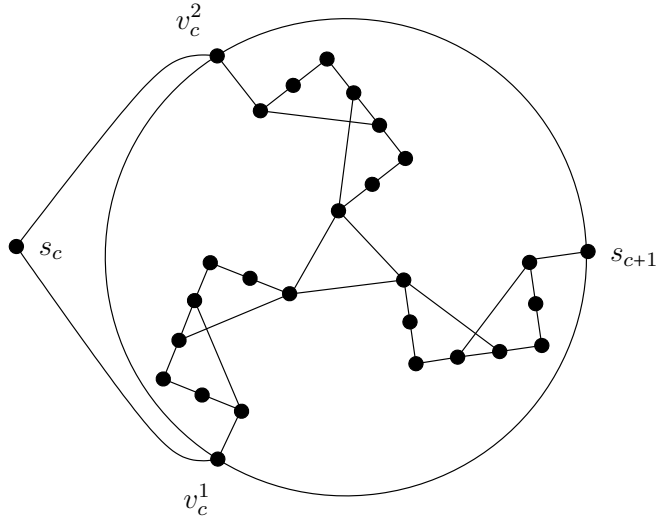}
\end{center}
\caption{ The graph $G^3_c$ corresponding to $x\oplus y \oplus z=0$.  }
\label{fig:12tsp:2}
\end{figure} 

Let $\cC_l$ be a circle in $\cH$ with variables $\{x^l_1, \ldots, x^l_n \}$.
For the circle border equation of $\cC_l$, we introduce 
the path $p_l=b^1_l-b^2_l-b^3_l$. In addition, we connect $b^3_l$ and $b^1_{l+1}$
to the parity graphs $P^l_1$ and $P^l_n$  in a similar way as in the
reduction from the Hybrid problem to the $(1,2)$--ATSP problem.
This is the whole description of the corresponding graph $G_{\cH}=(V(G_{\cH}),E(G_{\cH}))$.\\
\\
We are ready to give the proof of Theorem~\ref{thm:main12tsp}~$(i)$. 
 %
\subsection{Proof of Theorem~\ref{thm:main12tsp}~$(i)$}
Given $\cH$ an  instance of the Hybrid problem consisting of $n$ circles $\cC_1, \ldots, \cC_n$,
$m_2$ equations with two variables
and $m_3$ equations with three variables $g^3_c$ with $c\in [m_3]$,
 we construct in polynomial time 
the associated instance $G_{\cH}$ of the
$(1,2)$--TSP problem. 

Given an assignment $\phi$ to the variables of $\cH$ leaving $u$ equations unsatisfied in $\cH$, 
then, there is a tour 
with  length at most $8 \cdot m_2+ (3\cdot 8+3 )\cdot m_3 + 3n + 1$.

On the other hand, if we are given a tour $\sigma$ in $G_{\cH}$ with length 
$8 \cdot m_2+ (3\cdot 8+3 )\cdot m_3 + 3n + 1 $, it is 
possible to transform $\sigma$ in polynomial time into a tour $\sigma'$ such that it uses $0/1$-traversals
of all contained parity graphs in $G_{\cH}$ without increasing the length. Some cases are displayed in
 Figure~\ref{figcases12tsp}.
Moreover,  
we are able to construct in polynomial time an assignment to the variables of $\cH$, which leaves
at most $u$ equations in $\cH$ unsatisfied. \qed 

\section{Approximation Hardness of the $(1,4)$-TSP Problem}
In order to prove the claimed approximation hardness for the  
$(1,4)$-TSP problem, we cannot use the same parity graphs as in the construction
in the previous section since tours are not necessarily consistent in this metric. For this reason,
we introduce the parity graph depicted in Figure~\ref{fig:14tsp11} with the corresponding traversals. 

\begin{figure}[h]
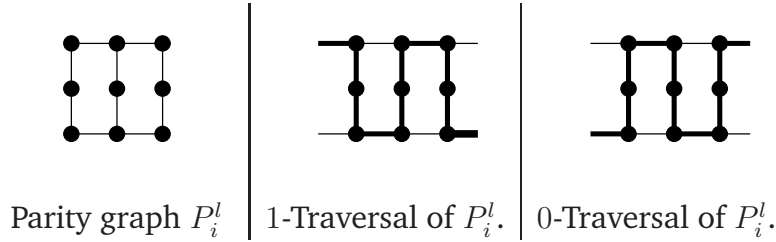

\begin{center}
\begin{tabular}[c]{c|c|c}
 \input{figures/fig14parity.pspdftex} &
  \input{figures/fig14tsp4.pspdftex} & \input{figures/fig14tsp4b.pspdftex} \\
Parity graph  $P^l_i$ & $1$-Traversal of $P^l_i$.   
& $0$-Traversal of $P^l_i$.
 \end{tabular}
\end{center}
\caption{$0/1$-Traversals of the graph $P^l_i$. }
\label{fig:14tsp11}
\end{figure} 
\begin{figure}[h]
\begin{center}
\input{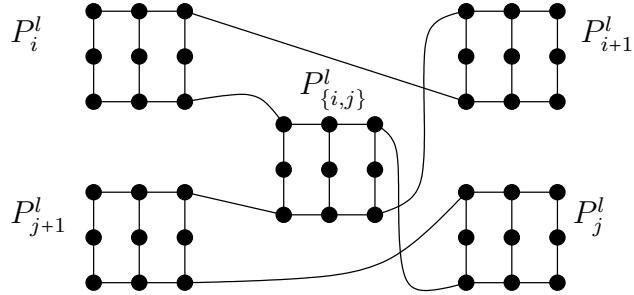}
\end{center}
\caption{ Graphs corresponding to $x^l_i \oplus x^l_j=0$, $x^l_i \oplus x^l_{i+1}=0$
and $x^l_j \oplus x^l_{j+1}=0$. }
\label{fig:14tsp:3}
\end{figure}
 
Given a matching equation $x^l_i \oplus x^l_j=0$ in $\cH$ and the circle equations
$x^l_i \oplus x^l_{i+1}=0$ and $x^l_j \oplus x^l_{j+1}=0$,
we connect the corresponding graphs as displayed in Figure~\ref{fig:14tsp:3}.\\
In order to define the new instance of the $(1,4)$--TSP problem, we replace 
all parity graphs in $G_{\cH}$ by graphs displayed in Figure~\ref{fig:14tsp11}.
In the remainder, we refer to this graph as $H_{\cH}$. All edges contained in 
a parity graph have weight $1$, whereas all other edges have weight $2$.
The remaining distances in the associated metric space $V_{\cH}$ are 
induced by the graphical metric in $H_{\cH}$ bounded by the value $4$
meaning 
 $$d_{\cH}(\{x,y\})=\min\{\textrm{length of a shortest path from $x$ to $y$ in $H_{\cH}$}, ~4 \}.$$  

This is the whole description of the associated instance $(V_{\cH},d_{\cH})$ 
of the $(1,4)$--TSP problem. We are ready to give the proof of Theorem~\ref{thm:main12tsp} $(ii)$.


\subsection{Proof of Theorem~\ref{thm:main12tsp} $(ii)$ }
Given $\cH$ an  instance of the Hybrid problem consisting of $n$ circles $\cC_1, \ldots, \cC_n$,
$m_2$ equations with two variables
and $m_3$ equations with three variables $g^3_c$ with $c\in [m_3]$,
 we construct in polynomial time 
the associated instance $(V_{\cH}, d_{\cH})$ of the
$(1,4)$--TSP problem.

Given an assignment $\phi$ to the variables of $\cH$ leaving $u$ equations unsatisfied in $\cH$, 
 there is a tour in $V_{\cH}$ with length at most  
 $ m_2\cdot (2+ 8) + m_3\cdot ( 3\cdot 10+2\cdot 3 )+ 6n+2  + 2\cdot u   $.
  
 On the other hand, if we are given a tour $\sigma$ in $(V_{\cH}, d_{\cH})$ with 
length $10m_2+36m_3+6n+2  +2\cdot u $, it is 
possible to transform $\sigma$ in polynomial time into a tour $\sigma'$ such that it uses $0/1$-traversals
of all contained parity graphs in $G_{\cH}$ without increasing the length. 
Some cases are displayed in 
Figure~\ref{figcases14tsp}. Then, we are 
able to construct in polynomial time an assignment to the variables of $\cH$, which leaves
at most $u$ equations in $\cH$ unsatisfied. \qed  
 
\section{Further Research}
We have improved (modestly) the best known approximation lower bounds for TSP 
with bounded metrics problems. They almost match the best known bounds
for the general (unbounded) metrics TSP problems. Because of a lack
of the good definability properties of metric TSP, further improvements
seem to be very difficult. A new less \textit{local}  method seems now to be
necessary to achieve much stronger results.

\newpage

\section{Figure Appendix }

\begin{figure}[h]
\begin{center}
\begin{tabular}[t]{c|c}
\input{figures/fig12tcasea.pspdftex} &  \input{figures/fig12tcaseb.pspdftex}  \\
$(a)$ & $(b)$ \\
\hline \\
\input{figures/fig12tcased.pspdftex} &  \input{figures/fig12tcasec.pspdftex}  \\
$(c)$ & $(d)$ \\
\hline \\
\input{figures/fig12tcasee.pspdftex} &  \input{figures/fig12tcasef.pspdftex}  \\
$(e)$ & $(f)$\\
\hline\\
\input{figures/fig12tcaseg.pspdftex} &  \input{figures/fig12tcaseh.pspdftex}  \\
$(g)$ & $(h)$
\end{tabular}
\end{center}
\caption{ Transformations yielding a consistent tour.}
\label{figcases12tsp}
\end{figure}

\newpage

\begin{figure}[h]
\begin{center}
\begin{tabular}[t]{c|c}
\input{figures/fig12acase1a.pspdftex} &  \input{figures/fig12acase1b.pspdftex}  \\
$(a)$ & $(b)$ \\
\hline \\
\input{figures/fig12acase1c.pspdftex} &  \input{figures/fig12acase1d.pspdftex}  \\
$(c)$ & $(d)$ \\
\hline \\
\input{figures/fig12acase1e.pspdftex} &  \input{figures/fig12acase1f.pspdftex}  \\
$(e)$ & $(f)$\\
\hline\\
\input{figures/fig12acase1g.pspdftex} &  \input{figures/fig12acase1h.pspdftex}  \\
$(g)$ & $(h)$
\end{tabular}
\end{center}
\caption{Transformations yielding a consistent tour. }
\label{figcases12atsp}
\end{figure}

\newpage

\begin{figure}
\begin{center}
\begin{tabular}[t]{c|c}
\input{figures/fig14atsp1a.pspdftex} &  \input{figures/fig14atsp1b.pspdftex}  \\
$(a)$ & $(b)$ \\
\hline \\
\input{figures/fig14atsp1c.pspdftex} &  \input{figures/fig14atsp1d.pspdftex}  \\
$(c)$ & $(d)$ \\
\hline \\
\input{figures/fig14atsp43a.pspdftex} &  \input{figures/fig14atsp43b.pspdftex}  \\
$(e)$ & $(f)$\\
\hline\\
\input{figures/fig14atsp41a.pspdftex} &  \input{figures/fig14atsp41b.pspdftex}  \\
$(g)$ & $(h)$
\end{tabular}
\end{center}
\caption{Transformations yielding a consistent tour. }
\label{figcases14atsp1}
\end{figure}

\newpage

\begin{figure}
\begin{center}
\begin{tabular}[t]{c|c}
\input{figures/fig14atsp34a.pspdftex} &  \input{figures/fig14atsp34b.pspdftex}  \\
$(a)$ & $(b)$ \\
\hline \\
\input{figures/fig14atsp33a.pspdftex} &  \input{figures/fig14atsp33b.pspdftex}  \\
$(c)$ & $(d)$ \\
\hline \\
\input{figures/fig14atsp31a.pspdftex} &  \input{figures/fig14atsp31b.pspdftex}  \\
$(e)$ & $(f)$\\
\hline\\
\input{figures/fig14atsp14.pspdftex} &  \input{figures/fig14atsp14b.pspdftex}  \\
$(g)$ & $(h)$
\end{tabular}
\end{center}
\caption{Transformations yielding a consistent tour. }
\label{figcases14atsp2}
\end{figure}

\newpage

\begin{figure}
\begin{center}
\begin{tabular}[t]{c|c}
\input{figures/fig14atsp13a.pspdftex} &  \input{figures/fig14atsp13b.pspdftex}  \\
$(a)$ & $(b)$ \\
\hline \\
\input{figures/fig14atsp33a.pspdftex} &  \input{figures/fig14atsp33b.pspdftex}  \\
$(c)$ & $(d)$ \\
\hline \\
\input{figures/fig14atsp31a.pspdftex} &  \input{figures/fig14atsp31b.pspdftex}  \\
$(e)$ & $(f)$\\
\hline\\
\input{figures/fig14atsp14.pspdftex} &  \input{figures/fig14atsp14b.pspdftex}  \\
$(g)$ & $(h)$
\end{tabular}
\end{center}
\caption{Transformations yielding a consistent tour. }
\label{figcases14atsp3}
\end{figure}

\newpage
\begin{figure}
\begin{center}
\begin{tabular}[t]{c|c}
\input{figures/fig14tsp3a.pspdftex} &  \input{figures/fig14tsp3b.pspdftex}  \\
$(a)$ & $(b)$ \\
\hline \\
\input{figures/fig14tsp3c.pspdftex} &  \input{figures/fig14tsp3d.pspdftex}  \\
$(c)$ & $(d)$ \\
\hline \\
\input{figures/fig14tsp3e.pspdftex} &  \input{figures/fig14tsp3b.pspdftex}  \\
$(e)$ & $(f)$\\
\hline\\
\input{figures/fig14tsp3g.pspdftex} &  \input{figures/fig14tsp3h.pspdftex}  \\
$(g)$ & $(h)$
\end{tabular}
\end{center}
\caption{Transformations yielding a consistent tour. }
\label{figcases14tsp}
\end{figure}

\end{document}